\documentclass[journal,11pt,onecolumn,draftclsnofoot]{IEEEtran}
\usepackage{cite}
\usepackage{graphicx,color,epsfig,rotating}
\usepackage{amsfonts,amsmath,amssymb,bbm}
\usepackage{setspace}
\usepackage{algorithm}
\usepackage[algo2e]{algorithm2e} 
\usepackage{subcaption}
\usepackage{mdwtab}
\usepackage{placeins}
\usepackage{psfrag, graphicx}
\usepackage[latin1]{inputenc}
\usepackage{multirow}
\usepackage{stfloats}
\usepackage{tabularx} 
\usepackage{booktabs} 
\usepackage{url}
\usepackage{bm}
\usepackage{geometry}
\usepackage{pstricks}
 \allowdisplaybreaks

\long\def\comment#1{}

\newtheorem{theorem}{Theorem}
\newtheorem{lemma}{Lemma}
\newtheorem{remark}{Remark}
\makeatletter
\newcommand{\subalign}[1]{%
  \vcenter{%
    \Let@ \restore@math@cr \default@tag
    \baselineskip\fontdimen10 \scriptfont\tw@
    \advance\baselineskip\fontdimen12 \scriptfont\tw@
    \lineskip\thr@@\fontdimen8 \scriptfont\thr@@
    \lineskiplimit\lineskip
    \ialign{\hfil$\m@th\scriptstyle##$&$\m@th\scriptstyle{}##$\crcr
      #1\crcr
    }%
  }
}
\newcommand{\thickhline}{%
    \noalign {\ifnum 0=`}\fi \hrule height 1pt
    \futurelet \reserved@a \@xhline
}
\newcolumntype{"}{@{\hskip\tabcolsep\vrule width 1pt\hskip\tabcolsep}}

\makeatother

\newfont{\bbb}{msbm10 scaled 700}

\newfont{\bb}{msbm10 scaled 1100}

\newcommand{\ratreg}{rate region\xspace}

\newcommand{\assod}{associated\xspace}

\newcommand{\survg}{surviving\xspace}
\newcommand{\Survg}{Surviving\xspace}

\newcommand{\Rmk}{Remark\xspace}

\newcommand{\rd}{round\xspace}
\newcommand{\rds}{rounds\xspace}

\newcommand{\msgs}{messages\xspace}

\newcommand{\ie}{i.e.\xspace}

\newcommand{\secty}{security\xspace}
\newcommand{\Secty}{Security\xspace}

\newcommand{\af}{as follows\xspace}

\newcommand{\itic}{information-theoretic\xspace}

\newcommand{\ntwk}{network\xspace}

\newcommand{\comm}{communication\xspace}
\newcommand{\comms}{communications\xspace}
\newcommand{\Comm}{Communication\xspace}

\newcommand{\achvb}{achievable\xspace}

\newcommand{\Rc}{{\cal R}}

\newcommand{\eqdef}{\stackrel{\Delta}{=}}

\newcommand{\be}{\begin{equation}}
\newcommand{\ee}{\end{equation}}
\newcommand{\bea}{\begin{eqnarray}}
\newcommand{\eea}{\end{eqnarray}}

\let\mrm\mathrm
\let\tbf\textbf
\let\tit\textit

\begin{document}
\title{On the Fundamental Limits of Hierarchical Secure Aggregation with Dropout and Collusion Resilience}

\author{
Zhou Li,~\IEEEmembership{Member,~IEEE},
Yizhou Zhao,~\IEEEmembership{Member,~IEEE},
Xiang~Zhang,~\IEEEmembership{Member,~IEEE},
and Giuseppe Caire,~\IEEEmembership{Fellow,~IEEE}

\thanks{Z. Li is with the Guangxi Key Laboratory of Multimedia Communications and Network Technology, 
Guangxi University, Nanning 530004, China (e-mail: lizhou@gxu.edu.cn).}

\thanks{Y. Zhao is with the College of Electronic and Information Engineering, Southwest University, Chongqing, China (e-mail: onezhou@swu.edu.cn).}

\thanks{X. Zhang and G. Caire are with the Department of Electrical Engineering and Computer Science, Technical University of Berlin, 10623 Berlin, Germany (e-mail: \{xiang.zhang, caire\}@tu-berlin.de).
}

\thanks{\tit{Corresponding author:} Xiang Zhang.}
}

\maketitle
\IEEEpeerreviewmaketitle

\begin{abstract}
We study the fundamental communication limits of information-theoretic secure aggregation in a hierarchical network consisting of a server, multiple relays, and multiple users per relay. Communication proceeds over two rounds and two hops, and the system is subject to arbitrary user and relay dropouts. Up to $T$ users may collude with either the server or any single relay. The server aims to recover the sum of the inputs of all users that survive the first round, while learning no additional information beyond the aggregate sum and the inputs of the colluding users. Each relay, however, must learn nothing about the users' inputs except for the information revealed by the inputs of the colluding users under the same collusion model.

We introduce a four-dimensional rate tuple that captures the communication cost across rounds and hops. Under a delayed message availability model, we establish necessary and sufficient conditions for feasibility and fully characterize the optimal first-round communication rates. For the second round, we characterize the optimal user-to-relay rate and derive lower and upper bounds on the relay-to-server rate. While these bounds do not coincide in general, they are tight in certain  regimes of interest.
Our results reveal a sharp threshold phenomenon: secure aggregation is feasible if and only if the total number of surviving users across surviving relays exceeds the collusion threshold. Achievability is established via a vector linear coding scheme with carefully structured correlated randomness exhibiting MDS-like properties, ensuring correctness and information-theoretic security under all possible dropout patterns. Entropic converse bounds are also derived.
\end{abstract}



\allowdisplaybreaks
\section{Introduction}

The rapid growth of distributed data sources has enabled large-scale collaborative learning and computation paradigms, such as federated learning~\cite{konecny2016federated,pmlr-v54-mcmahan17a,KairouzFL,9084352,yang2018applied}. In such systems, a central server aggregates information from a large number of users without directly accessing their individual data. A fundamental problem arising in this context is \emph{secure aggregation}~\cite{bonawitz2016practical,bonawitz2017practical}. The information-theoretic formulation of this problem was first studied by Zhao and Sun~\cite{zhao2023secure}, where the server computes the sum of the users' inputs while learning no additional information about the individual contributions.

In many practical deployments, communication between users and the server is not direct. Instead, hierarchical architectures are commonly adopted, where intermediate nodes such as relays or edge servers assist in collecting and forwarding information~\cite{11316456,zhang_mshsa_online2025,li2025collusionresilienthierarchicalsecureaggregation}. In such settings, secure aggregation must be performed over a hierarchical network, giving rise to the problem of \emph{hierarchical secure aggregation} (HSA)\cite{zhang2024optimal, zhang2025fundamental,11154963,lu2024capacity,li2025collusionresilienthierarchicalsecureaggregation,weng2026resilient}. This user-relay-server structure naturally arises in large-scale wireless and edge-assisted systems, where multi-hop communication is common. The introduction of a relay layer fundamentally changes the information flow: messages are processed and forwarded across multiple hops, imposing additional structural constraints on encoding and decoding operations.

Secure aggregation protocols must also address two additional practical challenges. First, robustness to user \emph{dropouts}~\cite{9834981,so2022lightsecagg,jahani2022swiftagg,zhang2025secure,liu2022efficient}: due to unreliable links, device failures, or limited energy, some users may fail to transmit their messages, and their identities are unknown a priori. Second, security under \emph{collusion}~\cite{li2023weakly,li2025weakly,Zhang_Li_Wan_DSA,Li_Zhang_GroupwiseDSA,Li_Zhang_WeaklyDSA}: the server or a subset of users may collude to infer private information. In hierarchical systems, these challenges are compounded across multiple hops and rounds, creating potential additional avenues for information leakage. A secure aggregation protocol must thus guarantee correct input sum recovery and \itic security
under arbitrary dropout and collusion patterns.

Motivated by  the  above  considerations, we study a minimal yet representative \tit{two-round} \comm model for hierarchical secure aggregation with user and relay dropouts. 
The system consists of $U$ relays, each associated with $V$ users, resulting in a total of $UV$ users. Each User $(u,v)$ holds an input $W_{u,v}$ and an independent offline-generated secret key $Z_{u,v}$. Communication occurs through the relay layer over two rounds.
In the first round, each User $(u,v)$ sends a message to  the associated Relay $u$. Let $\mathcal{V}^{(1)}_u \subseteq [V]$ denote the surviving users under Relay $u$, and $\mathcal{U}^{(1)} \subseteq [U]$ the surviving relays. The server observes the surviving user set $\bigcup_{u \in \mathcal{U}^{(1)}} \mathcal{V}^{(1)}_u$. 
In the second round,
the  \survg users from the first \rd 
may also drop out. Let $\mathcal{V}^{(2)}_u \subseteq \mathcal{V}^{(1)}_u$ and  $\mathcal{U}^{(2)} \subseteq \mathcal{U}^{(1)}$  denote  the  \survg users and relays  in the second \rd. To enable correct recovery, the \survg users in the second \rd will send additional \msgs via their \assod relays.
After two \rds of \comms, the server aims to recover the sum of the inputs of the  surviving users in the first round, \ie, 
$\sum_{(u,v) \in \bigcup_{u \in \mathcal{U}^{(1)}} \mathcal{V}^{(1)}_u} W_{u,v}$.

Secure aggregation in this model must satisfy both the correctness and \secty constraints \af.

\textbf{Correctness:} For all possible surviving sets, the server can correctly recover the desired sum using the received messages.

\textbf{Security:}
\begin{itemize}
    \item \emph{Server security:} Even if the server colludes with up to $T$ users, it cannot infer any additional information about the input set  beyond the colluding users' inputs and the input sum.
    \item \emph{Relay security:} Any Relay $u$, colluding with up to $T$ users, cannot infer additional information about the entire input set beyond the inputs of those users.
\end{itemize}

The goal of this work is to characterize the optimal \comm rates over the two \rds and two \ntwk hops subject to correctness and \secty constraints. Specifically, we determine the minimum number of symbols that must be transmitted over each round and hop to securely compute one symbol of the sum. We establish a threshold phenomenon: if the number of surviving users does not exceed the collusion threshold, secure aggregation is infeasible. Otherwise, we provide matching lower and upper bounds for the first round and for the first hop in the second round, and derive gap bounds for the second hop in the second round.

Prior works have explored extensions of secure aggregation~\cite{jahani2023swiftagg+,zhang2024optimal,10806947,egger2024privateaggregationhierarchicalwireless,zhang2025fundamental,11195652,egger2023private,lu2024capacity}. However, a sharp information-theoretic characterization of hierarchical secure aggregation with simultaneous dropouts and dual security remains an open problem. 
Our achievability schemes leverage classical secure multiparty computation techniques, including linear secret sharing and masking~\cite{BGW, CCD, Chor_Kushilevitz, Kushilevitz_Rosen}, adapted to the hierarchical and dropout-constrained setting. Converse results rely on Shannon's secrecy framework~\cite{Shannon1949} and extend entropy-based impossibility techniques from private information retrieval~\cite{Sun_Jafar_SPIR, Jia_Sun_Jafar, guo2020information, cheng2020capacity, wang2018e, wang2019symmetric,zhang2021fundamental}.

Compared to existing secure aggregation protocols~\cite{user_held, aggregation, aggregation_log, aggregation_turbo, aggregation_fast, bonawitz2019federated, Choi_Sohn_Han_Moon, pillutla2019robust, xu2019hybridalpha, beguier2020safer, so2020byzantine, elkordy2020secure, guo2020secure, alexandru2020private, lia2020privacy, truong2020privacy}, our work differs in three key aspects: (i) information-theoretic perfect security guarantees, rather than relying on computational assumptions; (ii) explicit hierarchical multi-hop modeling with user and relay dropouts; and (iii) characterization of the fundamental communication limits through matching achievability and converse bounds.

\subsection{Summary of Contributions}
The main contributions of this paper are summarized \af:

\begin{itemize}
\item  We propose a two-round hierarchical secure aggregation model that simultaneously captures user and relay dropouts and dual information-theoretic security against collusion with up to $T$ users by either the server or any single relay. This model introduces a four-dimensional rate tuple that quantifies the communication cost across rounds and hops in hierarchical networks.

\item Under a delayed message availability model, we establish necessary and sufficient feasibility conditions and fully characterize the optimal \tit{first-round} communication rates under the correctness and information-theoretic security constraints.

\item For the \tit{second round}, we characterize the optimal user-to-relay \comm rate and derive information-theoretic lower and upper bounds on the optimal relay-to-server rate. Although these bounds do not coincide in general, they are tight in several regimes of interest, yielding a near-complete characterization of the optimal communication cost.

\item Our results uncover a sharp threshold phenomenon: secure aggregation is feasible if and only if the total number of surviving users across surviving relays  exceeds the collusion threshold $T$. Achievability is achieved via vector linear coding with structured  and correlated randomness for secret key generation, while the converse follows from entropic secrecy arguments.
\end{itemize}



\section{Problem Statement}
\label{sec:model}


Consider secure aggregation in a three-layer hierarchical network consisting of an aggregation server, an intermediate layer with $U\ge 2$ relays, and a bottom layer of $UV$ users, where each relay serves a disjoint cluster of $V$ users. The network operates over two hops: the server communicates with all relays, and each relay communicates with its associated users 
\begin{figure}[h]
\centering
\includegraphics{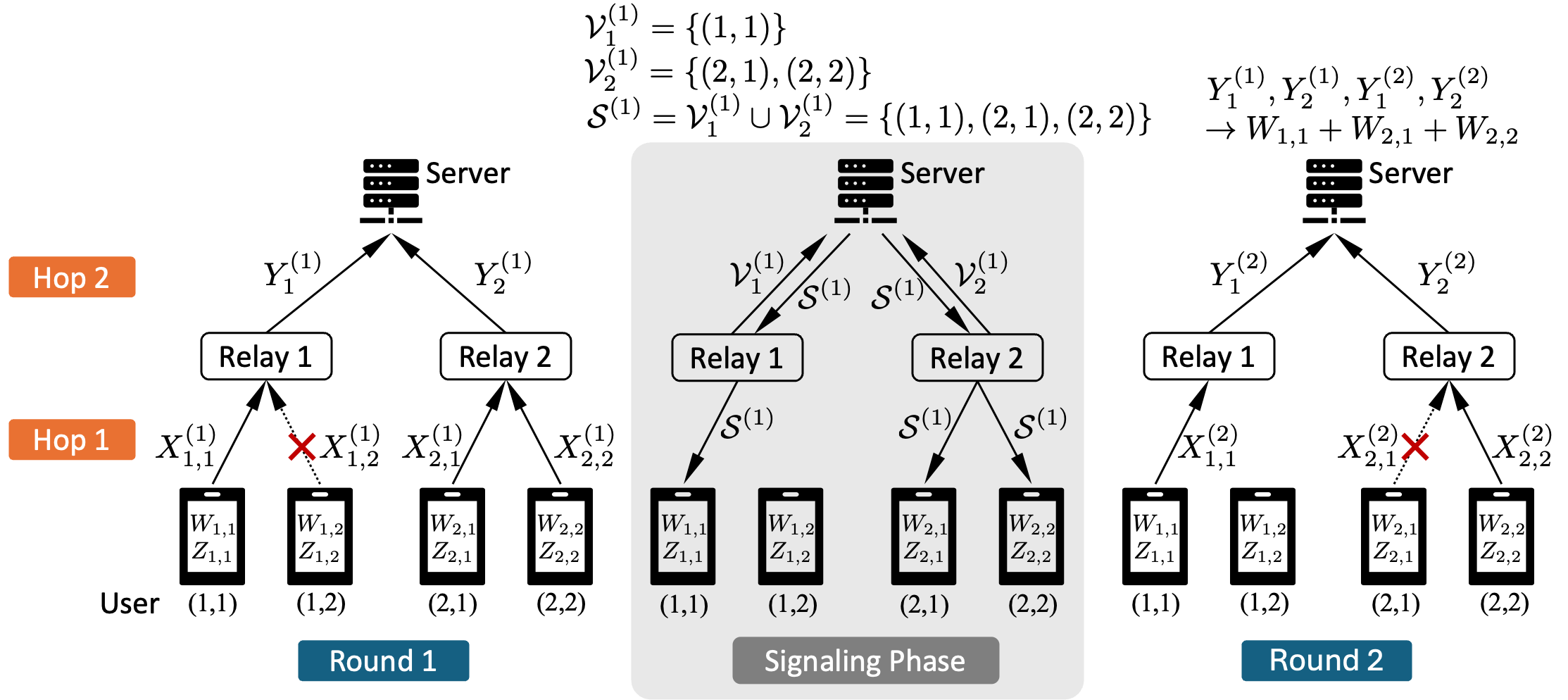}
\caption{\small Example of robust secure aggregation with $U=2$ relays and $UV=4$ users. 
In Round~1, User~$(1,2)$ drops out. 
During the signaling phase, the surviving relays report their surviving-user sets $\mathcal{V}^{(1)}$ to the server. 
The server then determines the first-round surviving-user set $\mathcal{S}^{(1)}$ and broadcasts it back to the surviving users via the surviving relays. 
This signaling phase is necessary because users must know the identities of the surviving users in the first round in order to generate subsequent messages. 
In Round~2, User~$(2,1)$ drops out. 
The server aims to securely compute $W_{1,1} + W_{2,1} + W_{2,2}$.}
\label{fig:model}
\end{figure}
(see Fig.~\ref{fig:model} for an example with $U=2$ and $V=2$). All communication links are assumed to be error-free. 
The $v^{\rm th}$ user associated with the $u^{\rm th}$ relay is labeled as $(u,v)\in[U]\times[V]$. Let
$\mathcal{M}_u \triangleq \{(u,v): v\in[V]\}$
denote the user cluster served by Relay $u$.
Each User $(u,v)$ possesses an input $W_{u,v}\in \mathbb{F}_q^{L}$, such as local gradients or model parameters in federated learning.
The user inputs $W_{[U]\times[V]} \triangleq \{W_{u,v}\}_{(u,v)\in[U]\times[V]}$ are assumed to be uniformly distributed  and mutually independent\footnote{The assumptions of input uniformity and independence are only used to establish the converse bounds. The proposed secure aggregation schemes ensures  security for arbitrarily distributed and correlated inputs.}. 
To protect the inputs,
each User $(u,v)$ is also equipped with a  secret key $Z_{u,v}$ of entropy $H(Z_{u,v})=L_Z$.
The individual keys $Z_{[U]\times[V]} \triangleq \{Z_{u,v}\}_{(u,v)\in[U]\times[V]}$ are generated from a \tit{source key} $Z_\Sigma \in \mathbb{F}_q^{L_{Z_\Sigma}}$ such that
$H(Z_{[U]\times[V]} \mid Z_\Sigma)=0$.\footnote{We assume the existence of a trusted third-party entity responsible for generating and distributing the individual keys to the users.} The source key $Z_\Sigma$ is only available to the trusted key generator and is not accessible to any user, relay, or the server.
Moreover, the keys are independent of the user inputs, i.e.,
\begin{align}
H\left(Z_{[U]\times[V]}, W_{[U]\times[V]}\right)
 & = H\left(Z_{[U]\times[V]}\right) + \sum_{u\in[U],v\in[V]} H(W_{u,v}), \label{independent}\\
 H(W_{u,v}) & = L \quad \text{(in $q$-ary units)}, \ \forall (u,v)\in[U]\times[V]. \label{inputsize}
\end{align}

\subsection{\Comm Protocol}
\label{subsec:comm protocol,problem formu}

The communication takes place over \tit{two rounds}, each consisting of two hops.

\tbf{1) First \rd.}
In the first round, over the first hop, User $(u,v)$ transmits a message $X^{(1)}_{u,v}$ to its associated Relay $u$.
The message $X^{(1)}_{u,v}$ consists of $L_X^{(1)}$ symbols over the finite field $\mathbb{F}_q$ and is generated as a deterministic function of the local input $W_{u,v}$ and the secret key $Z_{u,v}$, i.e.,
\begin{eqnarray}
H\left(X^{(1)}_{u,v} \mid W_{u,v}, Z_{u,v}\right)=0,
\quad \forall (u,v)\in[U]\times[V].
\label{messageX1}
\end{eqnarray}
After the first hop of the first round, an arbitrary subset of users may drop out.
For Relay $u$, $u\in[U]$, let $\mathcal{V}^{(1)}_u \subset \mathcal{M}_u$
denote the set of surviving users after the first hop in the first round, which can be an arbitrary subset of cardinality at least $V_{0}$, i.e.,
$|\mathcal{V}^{(1)}_u| \ge V_{0}$, where $1 \le V_{0} \le V-1$.

Over the second hop, Relay $u\in[U]$ generates and sends a message $Y^{(1)}_u$ to the aggregation server as a deterministic function of the messages $\{X^{(1)}_{u,v}\}_{(u,v)\in\mathcal{V}^{(1)}_u}$ received from the surviving users.
The message $Y^{(1)}_u$ consists of $L^{(1)}_Y$ symbols over the finite field $\mathbb{F}_q$. Thus,
\begin{eqnarray}
H\left(Y^{(1)}_u \mid \{X^{(1)}_{u,v}\}_{(u,v)\in\mathcal{V}^{(1)}_u}\right)=0,
\quad \forall u\in[U].
\label{messageY1}
\end{eqnarray}

Relays may also drop out after the second hop of the first round.
Let $\mathcal{U}^{(1)}$ denote the set of relays that remain active after the first round,
which can be any subset of $[U]$ with cardinality at least $U_{0}$, where $1 \le U_{0} \le U-1$.
The server then receives the messages $\{Y^{(1)}_u\}_{u\in\mathcal{U}^{(1)}}$ and aims to  compute
the sum of the inputs of \tit{all surviving users in the first \rd.}
For brevity of notation, denote
\be
\mrm{[\Survg\;\,users \;\,in \;\,1^{st} \;\,\rd]} \quad
\mathcal{S}^{(1)} \eqdef \bigcup_{u\in\mathcal{U}^{(1)}} \mathcal{V}^{(1)}_u
\ee
as the set of  the \survg  users in the first \rd.
The  desired sum is then equal to
$\sum_{(u,v)\in\mathcal{S}^{(1)}} W_{u,v}$.

\begin{remark}\tit{Each Relay $u \in [U]$ reports to the server the set of its surviving users $\mathcal{V}^{(1)}_{u}$. Based on the set of surviving relays $\mathcal{U}^{(1)}$ and the reported user sets $\{\mathcal{V}^{(1)}_{u}\}_{u \in \mathcal{U}^{(1)}}$, the server determines the complete set of surviving users in the first round,
$\bigcup_{u \in \mathcal{U}^{(1)}} \mathcal{V}^{(1)}_{u}$.
The server then broadcasts this set to all surviving relays, which forward it to their associated surviving users. This step informs each surviving user of the set of surviving users in the first round, enabling them to generate the appropriate second-round messages.}
\end{remark}

\tbf{2) Second \rd.}
In the second round, over the first hop, each User $(u,v)$ that is active 
at the beginning of the second round, i.e., each $(u,v) \in \mathcal{S}^{(1)}$, 
generates and transmits a message $X^{(2)}_{u,v}$,
which is a deterministic function of $W_{u,v}$ and $Z_{u,v}$,
and consists of $L^{(2)}_X$ symbols from $\mathbb{F}_q$\footnote{We consider the worst-case message length $L^{(2)}_X$ over all users and all first-round surviving user sets.}.  
\begin{eqnarray}
H\left(X_{u,v}^{(2)} \mid W_{u,v}, Z_{u,v}\right) = 0, \quad \forall (u,v) \in \mathcal{S}^{(1)}.
\label{messageX2}
\end{eqnarray}
After the first hop of the second round, an arbitrary subset of users may drop out.
For Relay $u\in[U]$, let $\mathcal{V}^{(2)}_u \subseteq \mathcal{V}^{(1)}_u$
denote the set of surviving users after the first hop in the second round.
The set $\mathcal{V}^{(2)}_u$ can be any subset satisfying
$|\mathcal{V}^{(2)}_u| \ge V_{0}$.

Over the second hop, each surviving Relay $u \in \mathcal{U}^{(1)}$ sends a message $Y^{(2)}_u$ to the aggregation server. After the second round, an arbitrary subset of relays may drop out.
Let $\mathcal{U}^{(2)} \subseteq \mathcal{U}^{(1)}$
denote the set of relays that remain active after the second round,
where we assume $|\mathcal{U}^{(2)}| \ge U_{0}$.
This message is a deterministic function of the messages sent by its surviving users:
\begin{eqnarray}
H\Big(Y^{(2)}_u \,\big|\, \{X_{u,v}^{(2)}\}_{(u,v)\in \mathcal{V}^{(2)}_u}\Big) &=& 0, 
\quad \forall u \in \mathcal{U}^{(1)}. \label{messageY2}
\end{eqnarray}
The message $Y^{(2)}_u$ consists of $L^{(2)}_Y$ symbols over the finite field $\mathbb{F}_q$.

\subsection{Correctness and \Secty}
\label{subsec:correctness and security,problem formu}

From the messages received from the surviving relays,
the aggregation server must be able to recover the desired sum $\sum_{(u,v)\in\mathcal{S}^{(1)}} W_{u,v}$ with zero error.
Formally, for any relay survival sets $\mathcal{U}^{(1)}$ and $\mathcal{U}^{(2)}$ satisfying
$\mathcal{U}^{(2)} \subseteq \mathcal{U}^{(1)} \subseteq [U]$ and $|\mathcal{U}^{(2)}| \ge U_{0}$,
and for any corresponding user survival sets $\mathcal{V}^{(1)}_u$ and $\mathcal{V}^{(2)}_u$ satisfying
$\mathcal{V}^{(2)}_u \subseteq \mathcal{V}^{(1)}_u \subseteq \{(u,v)\}_{v\in[V]}$ and
$|\mathcal{V}^{(2)}_u| \ge V_{0}$ for all $u \in \mathcal{U}^{(2)}$,
the following correctness constraint must hold:
\begin{eqnarray}
\mbox{[Correctness]}~~~ 
H\left(
\sum_{(u,v) \in \mathcal{S}^{(1)}} W_{u,v}
\;\Bigg|\;
\{Y^{(1)}_u\}_{u\in\mathcal{U}^{(1)}}, 
\{Y^{(2)}_u\}_{u\in\mathcal{U}^{(2)}}
\right) = 0.
\label{correctness}
\end{eqnarray}

The security constraints require that 
(i) \tit{Relay security:}  each relay should not obtain any information about the users' inputs $W_{[U]\times[V]}$, and 
(ii) \tit{Server security:} the aggregation server cannot learn any information about $W_{[U]\times[V]}$ beyond the desired sum 
$\sum_{(u,v)\in \mathcal{S}^{(1)}} W_{u,v}$, 
even if a relay or the server colludes with any set $\mathcal{T}$ of at most $T$ users.
For any Relay $u\in[U]$ and any colluding user set $\mathcal{T}$ with $|\mathcal{T}|\le T$, 
under the delayed message availability model (see \Rmk \ref{rmk:delayed msg model}), 
the relay security constraint is given by
\begin{align}
\mbox{[Relay security]}~~~I\Big(
\left\{X^{(1)}_{u,v}\right\}_{(u,v)\in \mathcal{M}_u},
\left\{X^{(2)}_{u,v}\right\}_{(u,v)\in \mathcal{V}^{(1)}_{u}};
\, W_{[U]\times [V]}
\;\Big|\;
\{W_{u,v}, Z_{u,v}\}_{(u,v)\in \mathcal{T}}
\Big)=0.
\label{relaysecurity}
\end{align}

\begin{remark}[Delayed Message Availability Model at the Relays]
\label{rmk:delayed msg model}
\tit{This corresponds to a worst-case adversarial model and leads to a strictly stronger security notion.
Throughout the security analysis, we adopt a delayed message availability model.
Specifically, a relay is assumed to eventually obtain all uplink messages transmitted
by users that were active at the beginning of each round, even if those users drop out
during the round.
As a result, in the first round Relay $u$ is assumed to observe
$\{X^{(1)}_{u,v}\}_{(u,v)\in\mathcal{M}_u}$,
and in the second round it is assumed to observe
$\{X^{(2)}_{u,v}\}_{(u,v)\in\mathcal{V}^{(1)}_u}$.
This assumption leads to a strictly stronger relay security requirement.
This model corresponds to a worst-case timing assumption on message delivery and dropout events.}
\end{remark}

For any colluding user set $\mathcal{T}$ with $|\mathcal{T}|\le T$, 
the server security constraint requires that
\begin{equation}  
\begin{aligned}[t]
\mbox{[Server security]}~~~I\Bigg( \left\{Y^{(1)}_u\right\}_{u\in [U]}, \left\{Y^{(2)}_u\right\}_{u\in \mathcal{U}^{(1)}}; W_{[U]\times [V]} \Bigg|\sum_{(u,v)\in \mathcal{S}^{(1)}} W_{u,v}, \{W_{u,v}, Z_{u,v}\}_{(u,v)\in \mathcal{T}} \Bigg)=0. \label{serversecurity}
\end{aligned}
\end{equation}

\begin{remark}[Delayed Message Availability at the Server]
\tit{In the server security constraint, the aggregation server is assumed to have access to
all relay messages transmitted in the first round,
including those from relays that subsequently drop out,
as well as all second-round messages from relays that remain active after the first round.
This models a delayed message availability scenario in which messages from dropped relays
may be delayed and eventually become available to the server.
Under this assumption, the server is allowed to observe
$\{Y^{(1)}_u\}_{u\in[U]}$ and $\{Y^{(2)}_u\}_{u\in\mathcal{U}^{(1)}}$,
leading to a strictly stronger server security requirement.}
\end{remark}

\subsection{Communication Rates and Achievable Region}

The communication \emph{rate} characterizes the number of transmitted symbols per input symbol and is defined as
\begin{equation}
R^{(1)}_X \triangleq \frac{L^{(1)}_X}{L}, \quad
R^{(1)}_Y \triangleq \frac{L^{(1)}_Y}{L}, \quad
R^{(2)}_X \triangleq \frac{L^{(2)}_X}{L}, \quad
R^{(2)}_Y \triangleq \frac{L^{(2)}_Y}{L}.
\label{rate}
\end{equation}
Here, $R^{(1)}_X$ and $R^{(1)}_Y$ denote the message rates of the first and second hops in the first round, respectively, while $R^{(2)}_X$ and $R^{(2)}_Y$ denote the message rates of the first and second hops in the second round.
A rate tuple $(R^{(1)}_X, R^{(1)}_Y, R^{(2)}_X, R^{(2)}_Y)$ is said to be \emph{achievable} if there exists a secure aggregation scheme, i.e., a construction of the secret keys 
$\{Z_{u,v}\}_{(u,v)\in[U]\times[V]}$, the messages 
$\{X^{(1)}_{u,v}\}_{(u,v)\in[U]\times[V]}$, 
$\{Y^{(1)}_u\}_{u\in [U]}$, 
$\{X^{(2)}_{u,v}\}_{(u,v)\in\mathcal{V}^{(1)}_u}$ and 
$\{Y^{(2)}_u\}_{u\in\mathcal{U}^{(1)}}$, 
such that the correctness constraint \eqref{correctness} and  the  security constraints \eqref{relaysecurity}, \eqref{serversecurity} are satisfied.
The \tit{optimal} \ratreg, denoted by $\Rc^*$, is defined as the closure of all \achvb rate tuples.  
In addition, let  $R_X^{(1),*}, R_Y^{(1),*}, R_X^{(2),*}$, and $R_Y^{(2),*}$ denote the \tit{individually} minimum rates, \ie, $R_l^{(i),*}\eqdef \min\{R_l^{(i)}:(R^{(1)}_X, R^{(1)}_Y, R^{(2)}_X, R^{(2)}_Y)\in \Rc^*  \}, l \in \{X,Y\},i=1,2$.

\section{Main Results}
\label{sec:result}

\begin{theorem}
\label{thm:main}
For hierarchical secure aggregation with $U$ relays, $V$ users per relay, dropout thresholds $V_0$ and $U_0$, and collusion threshold $T$, the optimal \ratreg  $\Rc^*$ is given by
\begin{eqnarray}
\label{eq: optimal rate region, thm 1}
\mathcal{R}^* =\left\{ 
\begin{array}{ll}
&
\left\{\begin{array}{ll}
\left(R_X^{(1)}, R_Y^{(1)}, R_X^{(2)}, R_Y^{(2)}\right) :
R_X^{(1)} \ge 1   ,
R_Y^{(1)} \ge 1,
R_X^{(2)} \ge  \frac{V_{0}}{U_{0}V_{0} - T}, 
R_Y^{(2)} \ge R_Y^{(2),*}
\end{array}\right\},\quad  \textrm{if } U_{0}V_{0} > T \\
&~~~~~~~~~~~~ \emptyset,\quad   \hfill \textrm{if } U_{0}V_{0} \le T
\end{array}\right.
\end{eqnarray}
where 
\be
\label{eq:Ry2 bound,thm}
\frac{1}{U_{0}-\left\lfloor {T}/{V_{0}} \right\rfloor} \le R_Y^{(2), *} \le \frac{1}{U_{0} - {T}/{V_{0}}}.
\ee
\end{theorem}

We provide  intuitive explanations  to Theorem~\ref{thm:main} \af:
\begin{itemize}
\item[\tit{1)}] \tit{Infeasibility:}
When $T \ge U_0 V_0$, i.e., the minimum number of surviving users is no greater than the collusion threshold, the information-theoretic secure aggregation problem is infeasible. In this regime, the number of independent user contributions is insufficient to simultaneously provide the randomness required for security and the degrees of freedom required for correct decoding. Hence, the correctness constraint (\ref{correctness}) and the security constraints (\ref{relaysecurity})(\ref{serversecurity}) cannot be satisfied simultaneously.

\item[\tit{2)}] \tit{First round rates $R_X^{(1),*}=R_Y^{(1),*}=1$:}
In the first round, each user and relay essentially transmits its input in full. Since this round is used to collect raw information from all users before any dropout occurs, no coding gain is possible. Hence, the rates are equal to $1$.

\item[\tit{3)}] \tit{Second round user rate $R_X^{(2),*} = \frac{V_{0}}{U_{0}V_{0}-T}$:} 
In the second round, surviving users send enough information so that the server can reconstruct the sum of all surviving inputs while remaining secure against any collusion of at most $T$ users. 
The numerator $V_0$ represents the minimum number of users per relay that survive, and the denominator $U_0 V_0 - T$ corresponds to the effective number of independent contributions after accounting for colluding users. This choice ensures information-theoretic security while minimizing redundancy.



\item[\tit{4)}] \tit{Second round relay rate $R_Y^{(2),*}$ bounds:} 
Each relay aggregates its surviving users' \msgs\ and forwards them to the server. 
The bounds 
$\frac{1}{U_0 - \lfloor T/V_0 \rfloor} \le R_Y^{(2),*} \le \frac{1}{U_0 - T/V_0}$
consist of a converse lower bound $R_Y^{(2),*} \ge \frac{1}{U_0 - \lfloor T/V_0 \rfloor} $ and an achievable upper bound $R_Y^{(2),*} \le \frac{1}{U_0 - T/V_0}$. 
Note that the upper and lower bounds coincide when either $V_0$ divides $T$ or $T=0$ (no collusion), thus characterizing the \tit{optimal} rate  region in  such scenarios.

\end{itemize}

\section{Achievability Proof of Theorem \ref{thm:main}}
\label{sec:ach}
Prior to the formal proof of the general achievability, we demonstrate the approach through two illustrative examples. The method is relatively simple and is based on standard vector linear coding techniques.

\subsection{Example 1: $U = 2, V = 2, U_{0} = 2, V_{0} = 1, T = 0$}

Consider $UV = 4$ users. Each relay has at least $V_{0} = 1$ responding user, and the server receives responses from at least $U_{0} = 2$ relays. There is no collusion between the users and the server or the relays, i.e., $T = 0$.  
The input length is set to $L = U_{0}V_{0} - T = 2$, so that each message satisfies $W_{u,v} \in \mathbb{F}_q^{2 \times 1}$.

To specify the correlated randomness, generate 4 i.i.d. uniform random vectors
$\{N_{u,v}\}_{(u,v)\in[2]\times[2]}$, where $N_{u,v} \in \mathbb{F}_q^{2 \times 1}$. 
We construct linear combinations corresponding to the sums of all subsets of 
$\{N_{1,1}, N_{1,2}, N_{2,1}, N_{2,2}\}$ with cardinality at least $U_{0}V_{0} = 2$.
For each User $(u,v)$, define the randomness variable as
\begin{equation}
Z_{u,v} =
\Big( N_{u,v}, \big\{ N_{i,j}(1) + 2^{2u+v-3} N_{i,j}(2) \big\}_{(i,j)\in[2]\times[2]} \Big), \quad \forall (u,v)\in[2]\times[2]. \label{1keymessage1}
\end{equation}

\textbf{Example:} For $(u,v) = (2,1)$, we have
\begin{equation}
\begin{aligned}
Z_{2,1}
=
(& N_{2,1},\,
N_{1,1}(1) + 2^{2} N_{1,1}(2),\,
N_{1,2}(1) + 2^{2} N_{1,2}(2), \\
 & N_{2,1}(1) + 2^{2} N_{2,1}(2),\,
N_{2,2}(1) + 2^{2} N_{2,2}(2)
).
\end{aligned}
\end{equation}

This completes the construction of the correlated secret keys.
We next describe the message design over two communication rounds.

\textbf{First hop, first round:} each user transmits its input masked by local noise:
\begin{equation}
    X^{(1)}_{u,v} = W_{u,v} + N_{u,v}, \quad \forall (u,v)\in [2]\times[2]. \label{1X1}
\end{equation}

\textbf{Second hop, first round:} each relay aggregates messages from its surviving users $\mathcal{V}^{(1)}_u \subseteq \mathcal{M}_u$:
\begin{equation}
    Y^{(1)}_u = \sum_{(u,v)\in \mathcal{V}^{(1)}_u} \bigl(W_{u,v}+N_{u,v}\bigr), \quad u\in [2].
\end{equation}

\textbf{Example surviving user combinations:}
\begin{itemize}
    \item $\mathcal{V}^{(1)}_1 = \{(1,1)\}, \mathcal{V}^{(1)}_2 = \{(2,2)\}$: 
    \begin{eqnarray}
        Y_1^{(1)} = W_{1,1}+N_{1,1},\quad Y_2^{(1)} = W_{2,2}+N_{2,2}.\label{Y1ex1t1}
    \end{eqnarray}
    \item $\mathcal{V}^{(1)}_1 = \{(1,1),(1,2)\}, \mathcal{V}^{(1)}_2 = \{(2,1)\}$: 
    \begin{eqnarray}
        Y_1^{(1)} = W_{1,1}+W_{1,2}+N_{1,1}+N_{1,2},\quad Y_2^{(1)} = W_{2,1}+N_{2,1}.\label{Y1ex1t2}
    \end{eqnarray}
    \item $\mathcal{V}^{(1)}_1 = \{(1,1),(1,2)\}, \mathcal{V}^{(1)}_2 = \{(2,1),(2,2)\}$: 
    \begin{eqnarray}
        Y_1^{(1)} = W_{1,1}+W_{1,2}+N_{1,1}+N_{1,2},\quad Y_2^{(1)} = W_{2,1}+W_{2,2}+N_{2,1}+N_{2,2}.\label{Y1ex1t3}
    \end{eqnarray}
\end{itemize}

After the first round, each Relay $u\in [2]$ reports to the server the set of its surviving users $\mathcal{V}^{(1)}_u$.
Based on the set of surviving relays $\mathcal{U}^{(1)}$ and the reported user sets
$\{\mathcal{V}^{(1)}_u\}_{u\in\mathcal{U}^{(1)}}$,
the server determines the complete set of users that survive the first round, denoted by $\mathcal{S}^{(1)} \triangleq \cup_{u\in\mathcal{U}^{(1)}} \mathcal{V}^{(1)}_u$.
The server then broadcasts $\mathcal{S}^{(1)}$ to all surviving relays, which forward it to their associated surviving users and request the transmission of the second-round messages.
Since $U = U_{0} = 2$, we have $\mathcal{U}^{(1)} = \mathcal{U}^{(2)} = \{1,2\}$.



\textbf{First hop, second round:} each surviving User $(u,v)\in \mathcal{S}^{(1)}$ transmits a linear combination of all surviving users' randomness:
\begin{equation}
    X^{(2)}_{u,v} = \sum_{(i,j)\in \mathcal{S}^{(1)}} \Big(N_{i,j}(1) + 2^{2u+v-3} N_{i,j}(2)\Big).
\end{equation}

\textbf{Example:}  
\begin{itemize}
    \item $\mathcal{S}^{(1)} = \{(1,1),(2,2)\}$:
    \begin{align}
        X^{(2)}_{1,1} &= N_{1,1}(1)+N_{1,1}(2) + N_{2,2}(1)+N_{2,2}(2),\notag\\
        X^{(2)}_{2,2} &= N_{1,1}(1)+2^3 N_{1,1}(2) + N_{2,2}(1)+2^3 N_{2,2}(2).
    \end{align}
    \item $\mathcal{S}^{(1)} = \{(1,1),(1,2),(2,1)\}$:
    \begin{align}
        X^{(2)}_{1,1} &= N_{1,1}(1)+N_{1,1}(2) + N_{1,2}(1)+N_{1,2}(2) + N_{2,1}(1)+N_{2,1}(2),\notag\\
        X^{(2)}_{1,2} &= N_{1,1}(1)+2 N_{1,1}(2) + N_{1,2}(1)+2 N_{1,2}(2) + N_{2,1}(1)+2 N_{2,1}(2),\notag\\
        X^{(2)}_{2,1} &= N_{1,1}(1)+2^2 N_{1,1}(2) + N_{1,2}(1)+2^2 N_{1,2}(2) + N_{2,1}(1)+2^2 N_{2,1}(2).
    \end{align}
    \item $\mathcal{S}^{(1)} = \{(1,1),(1,2),(2,1),(2,2)\}$:
    \begin{equation}  
\begin{aligned}[t]
X^{(2)}_{1,1} &= N_{1,1}(1)+N_{1,1}(2) + N_{1,2}(1)+N_{1,2}(2) + N_{2,1}(1)+N_{2,1}(2) + N_{2,2}(1)+N_{2,2}(2),\\
X^{(2)}_{1,2} &= N_{1,1}(1)+2 N_{1,1}(2) + N_{1,2}(1)+2 N_{1,2}(2) + N_{2,1}(1)+2 N_{2,1}(2) + N_{2,2}(1)+2 N_{2,2}(2),\\
X^{(2)}_{2,1} &= N_{1,1}(1)+2^2 N_{1,1}(2) + N_{1,2}(1)+2^2 N_{1,2}(2) + N_{2,1}(1)+2^2 N_{2,1}(2) + N_{2,2}(1)+2^2 N_{2,2}(2),\\
X^{(2)}_{2,2} &= N_{1,1}(1)+2^3 N_{1,1}(2) + N_{1,2}(1)+2^3 N_{1,2}(2) + N_{2,1}(1)+2^3 N_{2,1}(2) + N_{2,2}(1)+2^3 N_{2,2}(2).
\end{aligned}
\end{equation}
\end{itemize}

\textbf{Second hop, second round:} each relay aggregates $V_{0}$ messages from its surviving users $\mathcal{V}^{(2)}_u$. Let $\mathcal{Q}_u\subseteq \mathcal{V}^{(2)}_u$, $|\mathcal{Q}_u|=V_{0}$, then
\begin{equation}
    Y^{(2)}_u = \{ X^{(2)}_{u,v} \}_{v\in \mathcal{Q}_u}, \quad u\in [2].  \label{ex1Y2}
\end{equation}

\textbf{Example:}  
\begin{itemize}
    \item $\mathcal{S}^{(1)} = \{(1,1),(2,2)\}$, $\mathcal{V}^{(2)}_1 = \{(1,1)\}, \mathcal{V}^{(2)}_2 = \{(2,2)\}$: 
     \begin{align}
    &Y^{(2)}_1=\{X^{(2)}_{1,1}\}=\{N_{1,1}(1)+N_{1,1}(2) + N_{2,2}(1)+N_{2,2}(2)\},\notag\\
    &Y^{(2)}_2=\{X^{(2)}_{2,2}\}=\{N_{1,1}(1)+2^3 N_{1,1}(2) + N_{2,2}(1)+2^3 N_{2,2}(2)\}.\label{Y2ex1t1}
    \end{align}
    \item $\mathcal{S}^{(1)} = \{(1,1),(1,2),(2,1)\}$, $\mathcal{V}^{(2)}_1 = \{(1,1),(1,2)\}, \mathcal{V}^{(2)}_2 = \{(2,1)\}$: 
    \begin{align}
    &Y^{(2)}_1=\{X^{(2)}_{1,1}\}=\{N_{1,1}(1)+N_{1,1}(2) + N_{1,2}(1)+N_{1,2}(2) + N_{2,1}(1)+N_{2,1}(2)\},\notag\\
    &Y^{(2)}_2=\{X^{(2)}_{2,1}\}=\{N_{1,1}(1)+2^2 N_{1,1}(2) + N_{1,2}(1)+2^2 N_{1,2}(2) + N_{2,1}(1)+2^2 N_{2,1}(2)\}.\label{Y2ex1t2}
    \end{align}
    \item $\mathcal{S}^{(1)} = \{(1,1),(1,2),(2,1),(2,2)\}$, $\mathcal{V}^{(2)}_1 = \{(1,1),(1,2)\}, \mathcal{V}^{(2)}_2 = \{(2,1),(2,2)\}$: 
         \begin{align}
    Y^{(2)}_1=&\{X^{(2)}_{1,1}\}=\{N_{1,1}(1)+N_{1,1}(2) + N_{1,2}(1)+N_{1,2}(2)\notag\\
    & + N_{2,1}(1)+N_{2,1}(2)+ N_{2,2}(1)+N_{2,2}(2)\},\notag\\
    Y^{(2)}_2=&\{X^{(2)}_{2,1}\}=\{N_{1,1}(1)+2^2 N_{1,1}(2) + N_{1,2}(1)+2^2 N_{1,2}(2) \notag\\
    &+ N_{2,1}(1)+2^2 N_{2,1}(2)+ N_{2,2}(1)+2^2 N_{2,2}(2)\}.\label{Y2ex1t3}
    \end{align}
\end{itemize}

Based on the above construction, we next analyze the achievable rates of the scheme and show that, under these rates, both correctness and security requirements are satisfied.

\emph{Rate:}
Finally, we specify the communication rates of the proposed scheme.
The input length of each user is $L=2$ symbols over $\mathbb{F}_q$.
In the first round, each user transmits $L_X^{(1)}=2$ symbols and each relay
forwards $L_Y^{(1)}=2$ symbols. Hence, the first-round rates are
$R_X^{(1)} = \frac{L_X^{(1)}}{L} = 1,
R_Y^{(1)} = \frac{L_Y^{(1)}}{L} = 1.$
In the second round, each user transmits one symbol and each relay also
forwards one symbol, i.e., $L_X^{(2)}=L_Y^{(2)}=1$.
Therefore, the second-round rates are
$R_X^{(2)} = \frac{L_X^{(2)}}{L} = \frac{1}{2},
R_Y^{(2)} = \frac{L_Y^{(2)}}{L} = \frac{1}{2}.$
Equivalently, these rates can be expressed in terms of the system parameters as
$R_X^{(2)} = \frac{1}{U_{0}V_{0}-T},
R_Y^{(2)} = \frac{1}{U_{0}- \frac{T}{V_{0}}} ,$
which evaluate to $1/2$ for the considered setting $(U_{0},V_{0},T)=(2,1,0)$.
With these rates, the proposed scheme satisfies both the correctness
and the security requirements, as shown below.

\emph{Correctness:}
Let $\mathcal{S}^{(1)} = \mathcal{V}^{(1)}_1 \cup \mathcal{V}^{(1)}_2$ denote the set of surviving users after the first hop of the first round.  
Let $\mathcal{V}^{(2)}_u \subseteq \mathcal{V}^{(1)}_u \subseteq \mathcal{S}^{(1)}$ denote the surviving users in Relay $u$ for the second round. Since $V_{0}=1$ and $U_{0}=2$, we have:
$|\mathcal{V}^{(1)}_u| \ge V_{0}, |\mathcal{V}^{(2)}_u| \ge V_{0}, |\mathcal{U}^{(1)}| = |\mathcal{U}^{(2)}| = U_{0},  \forall u\in [2]$. At least $U_{0}$ relay surviving and each relay can always select $V_{0}$ messages from the users to send.
From the second round messages $\{Y^{(2)}_u\}_{u\in [2]}$ (see (\ref{ex1Y2})), the server can recover the aggregate key
$\sum_{(i,j)\in \mathcal{S}^{(1)}} N_{i,j} =$ $ \bigl(\sum_{(i,j)\in \mathcal{S}^{(1)}} N_{i,j}(1), \sum_{(i,j)\in \mathcal{S}^{(1)}} N_{i,j}(2)\bigr)$
with no error, because the coefficients in (\ref{1keymessage1}) satisfy the MDS property.  
Combining with the sum of the first round messages
$Y^{(1)}_1 + Y^{(1)}_2 = \sum_{(i,j)\in \mathcal{S}^{(1)}} \bigl(W_{i,j} + N_{i,j}\bigr),$
the server can decode the desired sum
$\sum_{(i,j)\in \mathcal{S}^{(1)}} W_{i,j}$
with no error.  
Thus, the scheme is correct for all possible user dropouts in both rounds.

\textbf{Example cases:}  

\begin{itemize}
    \item $\mathcal{S}^{(1)} = \{(1,1),(2,1)\}$: Users $(1,2)$ and $(2,2)$ drop in the first round. The server recovers $N_{1,1}+N_{2,1}=(N_{1,1}(1)+N_{2,1}(1),N_{1,1}(2)+N_{2,1}(2))$ from $Y^{(2)}_1,Y^{(2)}_2$ (see (\ref{Y2ex1t1})) and then $W_{1,1}+W_{2,1}=Y^{(1)}_1+Y^{(1)}_2-(N_{1,1}+N_{2,1})$ (see (\ref{Y1ex1t1})).
    
    \item $\mathcal{S}^{(1)} = \{(1,1),(1,2),(2,1)\}$: User $(2,2)$ drops. If one user per relay drops in the second round, the server still recovers $N_{1,1}+N_{1,2}+N_{2,1}=(N_{1,1}(1)+N_{1,2}(1)+N_{2,1}(1),N_{1,1}(2)+N_{1,2}(2)+N_{2,1}(2))$ (see (\ref{Y2ex1t1})) and then $W_{1,1}+W_{1,2}+W_{2,1}=Y^{(1)}_1+Y^{(1)}_2-(N_{1,1}+N_{1,2}+N_{2,1})$ (see (\ref{Y1ex1t2})).
    
    \item $\mathcal{S}^{(1)} = \{(1,1),(1,2),(2,1),(2,2)\}$: No users drop in the first round. The server recovers $N_{1,1}+N_{1,2}+N_{2,1}+N_{2,2}=(N_{1,1}(1)+N_{1,2}(1)+N_{2,1}(1)+N_{2,2}(1),N_{1,1}(2)+N_{1,2}(2)+N_{2,1}(2)+N_{2,2}(2))$ (see (\ref{Y2ex1t3})) and then $W_{1,1}+W_{1,2}+W_{2,1}+W_{2,2}=Y^{(1)}_1+Y^{(1)}_2-(N_{1,1}+N_{1,2}+N_{2,1}+N_{2,2})$ (see (\ref{Y1ex1t3})).
\end{itemize}

\emph{Security:}
The security of the proposed scheme relies on the following intuition.
In the first round, each user message is masked by independent randomness,
which guarantees information-theoretic security.
In the second round, the relays transmit carefully designed linear combinations
of the randomness symbols, which provide exactly the amount of side information
needed to recover the desired sum, and no more.

To verify the security constraints, it suffices to consider one representative
admissible realization.
Specifically, we consider the case
$\mathcal{S}^{(1)}=\{(1,1),(1,2),(2,1),(2,2)\},
\mathcal{V}^{(2)}_1 = \{(1,1),(1,2)\},
\mathcal{V}^{(2)}_2 = \{(2,1),(2,2)\}.$
We show that both the relay security constraint~\eqref{relaysecurity}
and the server security constraint~\eqref{serversecurity} are satisfied.

\medskip
\noindent\textbf{Relay security:}
Consider relay~1, we have
\begin{equation}
    \begin{aligned}
    &I(X^{(1)}_{1,1},X^{(1)}_{1,2},X^{(2)}_{1,1},X^{(2)}_{1,2};
    W_{1,1}, W_{1,2}, W_{2,1}, W_{2,2})\\
    =&H(X^{(1)}_{1,1},X^{(1)}_{1,2},X^{(2)}_{1,1},X^{(2)}_{1,2})-H(X^{(1)}_{1,1},X^{(1)}_{1,2},X^{(2)}_{1,1},X^{(2)}_{1,2}
    \mid W_{1,1}, W_{1,2}, W_{2,1}, W_{2,2})\\
    \leq&6 -H(N_{1,1},N_{1,2},
    N_{2,1}(1)+N_{2,1}(2) + N_{2,2}(1)+N_{2,2}(2),\notag\\
    &N_{2,1}(1)+2 N_{2,1}(2) + N_{2,2}(1)+2 N_{2,2}(2)
    \mid W_{1,1}, W_{1,2}, W_{2,1}, W_{2,2})\\
    =&6-6=0.
\end{aligned}
\end{equation}
Hence, relay~1 obtains no information about the users' messages, and the relay
security constraint is satisfied. The proofs for the other relays follow similarly.

\medskip
\noindent\textbf{Server security:}
The target sum that the server is allowed to learn is $W_{1,1}+W_{1,2}+W_{2,1}+W_{2,2}$. To verify server security, we consider the conditional mutual information:
\begin{align}
& I\left( W_{1,1}, W_{1,2}, W_{2,1}, W_{2,2};  Y^{(1)}_1, Y^{(1)}_2,Y^{(2)}_1,Y^{(2)}_2   \big| W_{1,1}+W_{1,2}+ W_{2,1}+ W_{2,2}\right) \notag \\
=& H( W_{1,1} + W_{1,2} + N_{1,1}+ N_{1,2},  W_{2,1} +W_{2,2} + N_{2,1}+ N_{2,2},\notag\\
&N_{1,1}(1)+N_{1,1}(2)+ N_{1,2}(1)+N_{1,2}(2)+ N_{2,1}(1)+N_{2,1}(2)+ N_{2,2}(1)+N_{2,2}(2),\notag\\
&N_{1,1}(1)+2^2N_{1,1}(2)+N_{1,2}(1)+2^2N_{1,2}(2)+ N_{2,1}(1)+2^2N_{2,1}(2)+ N_{2,2}(1)+2^2N_{2,2}(2)    \big|\notag\\
&W_{1,1}+W_{1,2}+ W_{2,1}+ W_{2,2}) \notag \\
& -H( N_{1,1}+ N_{1,2}, N_{2,1}+ N_{2,2}, N_{1,1}(1)+N_{1,2}(1)+ N_{2,1}(1)+ N_{2,2}(1), \notag\\
&N_{1,1}(2)+N_{1,2}(2)+ N_{2,1}(2)+ N_{2,2}(2)  \big| W_{1,1}, W_{1,2}, W_{2,1}, W_{2,2} ) \\
=& H( W_{1,1} + W_{1,2} + N_{1,1}+ N_{1,2},  W_{2,1} +W_{2,2} + N_{2,1}+ N_{2,2}   \big|W_{1,1}+W_{1,2}+ W_{2,1}+ W_{2,2}) \notag \\
& -H( N_{1,1}+ N_{1,2}, N_{2,1}+ N_{2,2}) \\
\leq&4-4=0.
\end{align}
Therefore, the server learns no additional information about the individual
messages beyond the desired sum, and the server security constraint is satisfied.

\medskip
In this representative case, both the relay security and server security
constraints hold.
Other admissible realizations differ only in the specific indices involved,
while the entropy structure and the independence of the randomness remain
unchanged. Hence, they can be analyzed in the same manner and lead to the same
conclusion.

\subsection{Example~2: $U = 3, V = 3, U_{0} = 2, V_{0} = 2, T = 2$}

We consider a system with $UV = 9$ users. Each relay receives messages from at least $V_0 = 2$ users, and the server receives messages from at least $U_0 = 2$ relays. Furthermore, at most $T = 2$ users may collude with the server or with any relay. We choose the input length as $L = U_0 V_0 - T = 2$.
Accordingly, each user input is given by
$W_{u,v} = \bigl(W_{u,v}(1), W_{u,v}(2)\bigr)^{\top} \in \mathbb{F}_q^{2 \times 1}$.

Next, we define the correlated randomness. Each User $(u,v)\in [3]\times [3]$ has two independent and uniformly distributed random vectors over $\mathbb{F}_q$: $N_{u,v} = \bigl(N_{u,v}(1), N_{u,v}(2)\bigr)^{\top} \in \mathbb{F}_q^{2 \times 1}$,
and
$S_{u,v} = \bigl(S_{u,v}(1), S_{u,v}(2)\bigr)^{\top} \in \mathbb{F}_q^{2 \times 1}$. In total, there are $2UV=18$ such vectors across all users.

From these local random vectors, we construct $UV = 9$ generic linearly independent linear combinations, which will be used to encode additional randomness accessible to each user. Specifically, for each User $(u,v)$, the available randomness is defined as
\begin{eqnarray}
Z_{u,v} &=& \Bigg( N_{u,v}, \Big\{ (N_{i,j}(1), N_{i,j}(2), S_{i,j}(1), S_{i,j}(2)) \cdot \boldsymbol{\alpha}_{u,v} \Big\}_{(i,j)\in[3]\times[3]} \Bigg), \label{ex2keys}
\end{eqnarray}
where $\boldsymbol{\alpha}_{u,v} \in \mathbb{F}_q^{4 \times 1}$ denotes the $(3(u-1)+v)$-th column of a $4 \times 9$ MDS matrix $\boldsymbol{\alpha}$. Here, the multiplication is a standard row-vector times column-vector product, producing a scalar for each $(i,j)$. In other words, we instantiate the abstract MDS column vector $\boldsymbol{\alpha}_{u,v}$ as
$\boldsymbol{\alpha}_{u,v} = (1, 2^{3(u-1)+v}, 3^{3(u-1)+v}, 4^{3(u-1)+v})^\top$
for the inner product computation.

This construction ensures that each user has access to both private and global linear combinations of the randomness, which is essential for achieving the desired security and recoverability properties in the system.
For example, when $u=1$ and $v=3$, the available randomness for user $(1,3)$ is
\begin{align}
Z_{1,3}&=(N_{1,3}, N_{1,1}(1)+2^{2}N_{1,1}(2)+3^{2}S_{1,1}(1)+4^{2}S_{1,1}(2),\notag\\
&N_{1,2}(1)+2^{2}N_{1,2}(2)+3^{2}S_{1,2}(1)+4^{2}S_{1,2}(2),N_{1,3}(1)+2^{2}N_{1,3}(2)+3^{2}S_{1,3}(1)+4^{2}S_{1,3}(2),\notag\\
&N_{2,3}(1)+2^{2}N_{2,3}(2)+3^{2}S_{2,3}(1)+4^{2}S_{2,3}(2),N_{3,1}(1)+2^{2}N_{3,1}(2)+3^{2}S_{3,1}(1)+4^{2}S_{3,1}(2),\notag\\
&N_{3,2}(1)+2^{2}N_{3,2}(2)+3^{2}S_{3,2}(1)+4^{2}S_{3,2}(2),N_{3,3}(1)+2^{2}N_{3,3}(2)+3^{2}S_{3,3}(1)+4^{2}S_{3,3}(2)).
\end{align}

We have now completed the design of the correlated secret keys. We next describe how the messages are constructed over the two communication rounds.

\textbf{First round, first hop:}
Each user transmits
\begin{align}
X^{(1)}_{u,v} = W_{u,v} + N_{u,v},
\quad \forall (u,v)\in[3]\times[3]. \label{ex2X1}
\end{align}

\textbf{First round, second hop:}
For any $\mathcal{V}^{(1)}_u\subseteq\{(u,v)\}_{v\in[3]}$ with
$|\mathcal{V}^{(1)}_u|\ge V_{0}=2$, Relay $u$ computes $Y_u^{(1)}$, and send to the server.
\begin{align}
Y_u^{(1)}
= \sum_{(u,v)\in\mathcal{V}^{(1)}_u}
\big(W_{u,v}+N_{u,v}\big),
\quad \forall u\in[3].
\end{align}
For example, if $\mathcal{S}^{(1)}=\{(1,1),(1,2),(2,1),(2,2),(2,3),(3,1),(3,3)\}$, we set
\begin{align}
    &Y_1^{(1)} =W_{1,1} +W_{1,2} + N_{1,1}+ N_{1,2},\notag\\
    &Y_2^{(1)} =W_{2,1} +W_{2,2} +W_{2,3} + N_{2,1}+ N_{2,2}+ N_{2,3},\notag\\
    &Y_3^{(1)} =W_{3,1}   +W_{3,3} + N_{3,1} + N_{3,3}.
\end{align}

After the first round, each relay reports to the server the set of its surviving users $\mathcal{V}^{(1)}_u$.
Based on the set of surviving relays $\mathcal{U}^{(1)}$ and the reported user sets
$\{\mathcal{V}^{(1)}_u\}_{u\in\mathcal{U}^{(1)}}$,
the server determines the complete set of users that survive the first round, denoted by $\mathcal{S}^{(1)} \triangleq \cup_{u\in\mathcal{U}^{(1)}} \mathcal{V}^{(1)}_u .$
The server then broadcasts $\mathcal{S}^{(1)}$ to all surviving relays, which forward it to their associated surviving users and request the transmission of the second-round messages.

\textbf{Second round, first hop:}
For each surviving User $(u,v)\in\mathcal{S}^{(1)}$, we set
\begin{align}
X^{(2)}_{u,v}
=& \sum_{(i,j)\in\mathcal{S}^{(1)}}
\Big( (N_{i,j}(1), N_{i,j}(2), S_{i,j}(1), S_{i,j}(2)) \cdot \boldsymbol{\alpha}_{u,v} 
\Big)\notag\\
=& \sum_{(i,j)\in\mathcal{S}^{(1)}}
\Big(
N_{i,j}(1)
+2^{3u+v-4}N_{i,j}(2)
+3^{3u+v-4}S_{i,j}(1)
+4^{3u+v-4}S_{i,j}(2)
\Big).
\end{align}
For example, if $\mathcal{S}^{(1)}={(1,1),(1,2),(2,1),(2,2),(2,3),(3,1),(3,3)}$ and for $(u,v)=(2,1)$, we set
\begin{flalign}
    X^{(2)}_{2,1}=&N_{1,1}(1)+2^3N_{1,1}(2)+3^3S_{1,1}(1)+4^3S_{1,1}(2)+N_{1,2}(1)+2^3N_{1,2}(2)+3^3S_{1,2}(1)+4^3S_{1,2}(2)+&\notag\\
    &N_{2,1}(1)+2^3N_{2,1}(2)+3^3S_{2,1}(1)+4^3S_{2,1}(2)+N_{2,2}(1)+2^3N_{2,2}(2)+3^3S_{2,2}(1)+4^3S_{2,2}(2)+&\notag\\
    &N_{2,3}(1)+2^3N_{2,3}(2)+3^3S_{2,3}(1)+4^3S_{2,3}(2)+ N_{3,1}(1)+2^3N_{3,1}(2)+3^3S_{3,1}(1)+4^3S_{3,1}(2)+&\notag\\
    &N_{3,3}(1)+2^3N_{3,3}(2)+3^3S_{3,3}(1)+4^3S_{3,3}(2)
\end{flalign}

\textbf{Second round, second hop:}
Since $|\mathcal{V}^{(2)}_u|\ge V_{0}$, each Relay $u$ selects a subset
$\mathcal{Q}_u\subseteq\mathcal{V}^{(2)}_u$ with $|\mathcal{Q}_u|=V_{0}$
and forwards the corresponding messages to the server. Specifically,
\begin{align}
Y^{(2)}_u = \{ X^{(2)}_{u,v} \}_{v\in\mathcal{Q}_u},
\quad u\in[3].
\end{align}
For example, if $\mathcal{V}^{(2)}_2=\{(2,1),(2,2),(2,3)\}$, we choose $\mathcal{Q}_{2}=\{(2,1),(2,2)\}$, and set
\begin{align}
    &Y^{(2)}_1=\{X^{(2)}_{1,1},X^{(2)}_{1,2}\},\notag\\
    &Y^{(2)}_2=\{X^{(2)}_{2,1},X^{(2)}_{2,2}\},\notag\\
    &Y^{(2)}_3=\{X^{(2)}_{3,1},X^{(2)}_{3,3}\}.
\end{align}

We next analyze the achievable rates and establish correctness and security.

\emph{Rate:}
Finally, we specify the communication rates of the proposed scheme.
The first-round rates are given by
$R_X^{(1)}=\frac{L_X^{(1)}}{L}=1,
R_Y^{(1)}=\frac{L_Y^{(1)}}{L}=1.$
For the second round, we have
$R_X^{(2)}=\frac{L_X^{(2)}}{L}=\frac{1}{2},
R_Y^{(2)}=\frac{L_Y^{(2)}}{L}=1.$
Equivalently, the second-round rates can be expressed in terms of the system
parameters as
$R_X^{(2)}=\frac{1}{U_{0}V_{0}-T},
R_Y^{(2)}=\frac{1}{U_{0}- \frac{T}{V_{0}}}.$
With these rates, the proposed scheme satisfies both the correctness and the
security requirements, as shown below.

\emph{Correctness:}
Let $\mathcal{S}^{(1)}=\mathcal{V}^{(1)}_1\cup\mathcal{V}^{(1)}_2\cup\mathcal{V}^{(1)}_3$ denote the set of surviving users after the first hop of the first round, and let $\mathcal{V}^{(2)}_u\subseteq\mathcal{V}^{(1)}_u\subseteq\mathcal{S}^{(1)}$ denote the set of surviving users at Relay $u$ in the second round. Since $V_{0}=2$ and $U_{0}=2$, it holds that $|\mathcal{V}^{(1)}_u|\ge V_{0}$, $|\mathcal{V}^{(2)}_u|\ge V_{0}$, and $|\mathcal{U}^{(1)}|,|\mathcal{U}^{(2)}|\geq U_{0}$ for all $u\in[3]$, which implies that at least $U_{0}$ relays survive and each surviving relay can always select $V_{0}$ user messages for transmission.
From the second round messages $\{Y^{(2)}_u\}_{u\in \mathcal{U}^{(2)}}$, the server can recover the aggregation keys
$\sum_{(i,j)\in \mathcal{S}^{(1)}} N_{i,j} =$ $ \bigl(\sum_{(i,j)\in \mathcal{S}^{(1)}} N_{i,j}(1),$ $ \sum_{(i,j)\in \mathcal{S}^{(1)}} N_{i,j}(2)\bigr)$ and $\sum_{(i,j)\in \mathcal{S}^{(1)}} S_{i,j} =$ $ \bigl(\sum_{(i,j)\in \mathcal{S}^{(1)}} S_{i,j}(1), \sum_{(i,j)\in \mathcal{S}^{(1)}} S_{i,j}(2)\bigr)$
with no error, because the precoding matrices in (\ref{ex2keys}) are MDS.  
Combining with the sum of the first round messages
$Y^{(1)}_1 + Y^{(1)}_2 = \sum_{(i,j)\in \mathcal{S}^{(1)}} \bigl(W_{i,j} + N_{i,j}\bigr),$
the server can decode the desired sum
$\sum_{(i,j)\in \mathcal{S}^{(1)}} W_{i,j}$.

\textbf{Example Case:}
Suppose
$\mathcal{V}^{(1)}_1 \cup \mathcal{V}^{(1)}_2 \cup \mathcal{V}^{(1)}_3
= \{(1,1),(1,2),(2,1),(2,2),(2,3),(3,1),(3,3)\},$
i.e., users $\{(1,3),$ $(3,2)\}$ drop out in the first hop of the first round.
Assume that $\mathcal{U}^{(1)}=\{1,2,3\}$, i.e., no relay drops out in the first
round, and that $\mathcal{U}^{(2)}=\{1,2\}$.
From $Y^{(2)}_1$ and $Y^{(2)}_2$, the server can recover $\sum_{(i,j)\in \mathcal{S}^{(1)}} N_{i,j} =$ $ \bigl(\sum_{(i,j)\in \mathcal{S}^{(1)}} N_{i,j}(1), \sum_{(i,j)\in \mathcal{S}^{(1)}} N_{i,j}(2)\bigr)$ and $\sum_{(i,j)\in \mathcal{S}^{(1)}} S_{i,j} =$ $ \bigl(\sum_{(i,j)\in \mathcal{S}^{(1)}} S_{i,j}(1), \sum_{(i,j)\in \mathcal{S}^{(1)}} S_{i,j}(2)\bigr)$ without error, again due to the MDS property of the precoding matrices in~\eqref{ex2X1}.
Combining these with $Y^{(1)}_1+Y^{(1)}_2 = \sum_{(u,v)\in\mathcal{S}^{(1)}} W_{u,v} + \sum_{(u,v)\in\mathcal{S}^{(1)}} N_{u,v},$
the desired sum $\sum_{(u,v)\in\mathcal{S}^{(1)}} W_{u,v}$ can be decoded with zero error.
The correctness proof for other admissible dropout patterns follows in the same manner.

{\em Security:} 
To verify the security constraints, it suffices to consider one representative
admissible realization.
Specifically, we consider the case
$\mathcal{S}^{(1)}=\{(1,1),(1,2),(2,1),(2,2),(2,3),(3,1),(3,3)\},$ $
\mathcal{V}^{(2)}_1 = \{(1,1),(1,2)\},
\mathcal{V}^{(2)}_2 = \{(2,1),(2,2)\}$, $\mathcal{U}^{(2)}=\{1,2\}$, and $\mathcal{T}=\{(1,1),(2,1)\}$.
We show that both the relay security constraint~\eqref{relaysecurity}
and the server security constraint~\eqref{serversecurity} are satisfied.

\noindent\textbf{Relay security:}
Consider relay~1. We have
\begin{align}    &I(X^{(1)}_{1,1},X^{(1)}_{1,2},X^{(1)}_{1,3},X^{(2)}_{1,1},X^{(2)}_{1,2};\{W_{u,v}\}_{(u,v)\in [3]\times [3]}|W_{1,1},Z_{1,1},W_{2,1},Z_{2,1})\\
    =&H(X^{(1)}_{1,1},X^{(1)}_{1,2},X^{(1)}_{1,3},X^{(2)}_{1,1},X^{(2)}_{1,2}|W_{1,1},Z_{1,1},W_{2,1},Z_{2,1})\notag\\
    &-H(X^{(1)}_{1,1},X^{(1)}_{1,2},X^{(1)}_{1,3},X^{(2)}_{1,1},X^{(2)}_{1,2}|\{W_{u,v}\}_{(u,v)\in [3]\times [3]}, Z_{1,1}, Z_{2,1})\\
    =&H(W_{1,2}+Z_{1,2},W_{1,3}+Z_{1,3},\notag\\
    &N_{1,1}(1)+2N_{1,1}(2)+3S_{1,1}(1)+4S_{1,1}(2)+N_{1,2}(1)+2N_{1,2}(2)+3S_{1,2}(1)+4S_{1,2}(2)+\notag\\
    &N_{2,1}(1)+2N_{2,1}(2)+3S_{2,1}(1)+4S_{2,1}(2)+N_{2,2}(1)+2N_{2,2}(2)+3S_{2,2}(1)+4S_{2,2}(2)+\notag\\
    &N_{2,3}(1)+2N_{2,3}(2)+3S_{2,3}(1)+4S_{2,3}(2)+ N_{3,1}(1)+2N_{3,1}(2)+3S_{3,1}(1)+4S_{3,1}(2)+\notag\\
    &N_{3,3}(1)+2N_{3,3}(2)+3S_{3,3}(1)+4S_{3,3}(2)|W_{1,1},Z_{1,1},W_{2,1},Z_{2,1})\notag\\
    &-H(Z_{1,2},Z_{1,3},\notag\\
    &N_{1,1}(1)+2N_{1,1}(2)+3S_{1,1}(1)+4S_{1,1}(2)+N_{1,2}(1)+2N_{1,2}(2)+3S_{1,2}(1)+4S_{1,2}(2)+\notag\\
    &N_{2,1}(1)+2N_{2,1}(2)+3S_{2,1}(1)+4S_{2,1}(2)+N_{2,2}(1)+2N_{2,2}(2)+3S_{2,2}(1)+4S_{2,2}(2)+\notag\\
    &N_{2,3}(1)+2N_{2,3}(2)+3S_{2,3}(1)+4S_{2,3}(2)+ N_{3,1}(1)+2N_{3,1}(2)+3S_{3,1}(1)+4S_{3,1}(2)+\notag\\
    &N_{3,3}(1)+2N_{3,3}(2)+3S_{3,3}(1)+4S_{3,3}(2)|\{W_{u,v}\}_{(u,v)\in [3]\times [3]}, Z_{1,1},Z_{2,1})\\
    =&5-5=0.
\end{align}
The equality follows since the two entropy terms involve the same set of independent linear combinations. Conditioning on all $W_{u,v}$ removes only the data components while preserving the independent randomness, and each remaining term contributes one $q$-ary symbol.
Hence, relay~1 obtains no information about the users' messages, and the relay security constraint is satisfied.

\medskip
\noindent\textbf{Server security:}
For server security, we have
\begin{align}
& I\big( \{W_{u,v}\}_{(u,v)\in [3]\times [3]};  Y^{(1)}_1, Y^{(1)}_2,Y^{(1)}_3,Y^{(2)}_1,Y^{(2)}_2,Y^{(2)}_3   \big| \notag\\
&W_{1,1}+W_{1,2}+W_{2,1}+W_{2,2}+W_{2,3}+W_{3,1}+W_{3,3}, W_{1,1},Z_{1,1}, W_{2,1}, Z_{2,1}\big) \notag \\
=& H( W_{1,1} + W_{1,2} + N_{1,1}+ N_{1,2},  W_{2,1} + W_{2,2} + W_{2,3} + N_{2,1} + N_{2,2} + N_{2,3}, \notag\\
&W_{3,1} + W_{3,3} + N_{3,1} + N_{3,3},\notag\\
&N_{1,1}(1)+N_{1,2}(1)+N_{2,1}(1)+N_{2,2}(1)+N_{2,3}(1)+N_{3,1}(1)+N_{3,3}(1),\notag\\
&N_{1,1}(2)+N_{1,2}(2)+N_{2,1}(2)+N_{2,2}(2)+N_{2,3}(2)+N_{3,1}(2)+N_{3,3}(2),\notag\\
&S_{1,1}(1)+S_{1,2}(1)+S_{2,1}(1)+S_{2,2}(1)+S_{2,3}(1)+S_{3,1}(1)+S_{3,3}(1),\notag\\
&S_{1,1}(2)+S_{1,2}(2)+S_{2,1}(2)+S_{2,2}(2)+S_{2,3}(2)+S_{3,1}(2)+S_{3,3}(2)\notag\\
    & \big|W_{1,1}+W_{1,2}+ W_{2,1}+ W_{2,2}+ W_{2,3}+ W_{3,1}+ W_{3,3},W_{1,1},Z_{1,1}, W_{2,1}, Z_{2,1}) \notag\\
    &-H( N_{1,1}+ N_{1,2},  N_{2,1} + N_{2,2}+ N_{2,3}, N_{3,1} + N_{3,3},\notag\\
    &N_{1,1}(1)+N_{1,2}(1)+N_{2,1}(1)+N_{2,2}(1)+N_{2,3}(1)+N_{3,1}(1)+N_{3,3}(1),\notag\\
&N_{1,1}(2)+N_{1,2}(2)+N_{2,1}(2)+N_{2,2}(2)+N_{2,3}(2)+N_{3,1}(2)+N_{3,3}(2),\notag\\
&S_{1,1}(1)+S_{1,2}(1)+S_{2,1}(1)+S_{2,2}(1)+S_{2,3}(1)+S_{3,1}(1)+S_{3,3}(1),\notag\\
&S_{1,1}(2)+S_{1,2}(2)+S_{2,1}(2)+S_{2,2}(2)+S_{2,3}(2)+S_{3,1}(2)+S_{3,3}(2)\big|\notag\\
&\{W_{u,v}\}_{(u,v)\in [3]\times [3]},Z_{1,1}, Z_{2,1})  \\
=&H( W_{1,2} + N_{1,2}, W_{2,2} + W_{2,3} + N_{2,2}+N_{2,3}, W_{3,1} + W_{3,3} + N_{3,1} + N_{3,3},\notag\\
&S_{1,2}(1)+S_{2,2}(1)+S_{2,3}(1)+S_{3,1}(1)+S_{3,3}(1),S_{1,2}(2)+S_{2,2}(2)+S_{2,3}(2)+S_{3,1}(2)+S_{3,3}(2) \big|\notag\\
&W_{1,1}+W_{1,2}+ W_{2,1}+ W_{2,3}+ W_{3,1}+ W_{3,3}, W_{1,1},Z_{1,1}, W_{2,1}, Z_{2,1})\notag\\
&-  H(N_{1,2},  N_{2,2}+N_{2,3}, N_{3,1} + N_{3,3}, S_{1,2}(1)+S_{2,2}(1)+S_{2,3}(1)+S_{3,1}(1)+S_{3,3}(1),\notag\\
    &S_{1,2}(2)+S_{2,2}(2)+S_{2,3}(2)+S_{3,1}(2)+S_{3,3}(2) \big|\{W_{u,v}\}_{(u,v)\in [3]\times [3]},Z_{1,1}, Z_{2,1})\label{ex2serversect1}  \\
    =& H( W_{1,2} + N_{1,2}, W_{2,2} + W_{2,3} + N_{2,2}+N_{2,3}, W_{3,1} + W_{3,3} + N_{3,1} + N_{3,3}\big|\notag\\
&W_{1,1}+W_{1,2}+ W_{2,1}+ W_{2,3}+ W_{3,1}+ W_{3,3}, W_{1,1},Z_{1,1}, W_{2,1}, Z_{2,1})\notag\\
&+H( S_{1,2}(1)+S_{2,2}(1)+S_{2,3}(1)+S_{3,1}(1)+S_{3,3}(1),\notag\\
&S_{1,2}(2)+S_{2,2}(2)+S_{2,3}(2)+S_{3,1}(2)+S_{3,3}(2) \big|\notag\\
&W_{1,2} + N_{1,2}, W_{2,2} + W_{2,3} + N_{2,2}+N_{2,3}, W_{3,1} + W_{3,3} + N_{3,1} + N_{3,3},\notag\\
&W_{1,1}+W_{1,2}+ W_{2,1}+ W_{2,3}+ W_{3,1}+ W_{3,3}, W_{1,1},Z_{1,1}, W_{2,1}, Z_{2,1})\notag\\
&-  H(N_{1,2},  N_{2,2}+N_{2,3}, N_{3,1} + N_{3,3} \big|\{W_{u,v}\}_{(u,v)\in [3]\times [3]},Z_{1,1}, Z_{2,1}) \notag \\
    &-  H( S_{1,2}(1)+S_{2,2}(1)+S_{2,3}(1)+S_{3,1}(1)+S_{3,3}(1),\notag\\
    &S_{1,2}(2)+S_{2,2}(2)+S_{2,3}(2)+S_{3,1}(2)+S_{3,3}(2) \big|\notag\\
    &N_{1,2},  N_{2,2}+N_{2,3}, N_{3,1} + N_{3,3},\{W_{u,v}\}_{(u,v)\in [3]\times [3]},Z_{1,1}, Z_{2,1}) \label{ex2serversect2} \\
   =&3\times 2-3\times 2=0,
\end{align}
where the first term of \eqref{ex2serversect1} holds because the quantities 
$N_{1,1}(1)+N_{1,2}(1)+N_{2,1}(1)+N_{2,2}(1)+N_{2,3}(1)+N_{3,1}(1)+N_{3,3}(1)$ and 
$N_{1,1}(2)+N_{1,2}(2)+N_{2,1}(2)+N_{2,2}(2)+N_{2,3}(2)+N_{3,1}(2)+N_{3,3}(2)$ 
are fully determined by the sums 
$W_{1,1}+W_{1,2}+N_{1,1}+N_{1,2}$, 
$W_{2,1}+W_{2,2}+W_{2,3}+N_{2,1}+N_{2,2}+N_{2,3}$, 
$W_{3,1}+W_{3,3}+N_{3,1}+N_{3,3}$, 
and the target sum 
$W_{1,1}+W_{1,2}+W_{2,1}+W_{2,3}+W_{3,1}+W_{3,3}$.
The second term of \eqref{ex2serversect2} is zero because the sums 
$S_{1,2}(1)+S_{2,2}(1)+S_{2,3}(1)+S_{3,1}(1)+S_{3,3}(1)$ and 
$S_{1,2}(2)+S_{2,2}(2)+S_{2,3}(2)+S_{3,1}(2)+S_{3,3}(2)$ 
are fully determined by 
$W_{1,2}+N_{1,2}$, 
$W_{2,2}+W_{2,3}+N_{2,2}+N_{2,3}$, 
$W_{3,1}+W_{3,3}+N_{3,1}+N_{3,3}$, 
the target sum $W_{1,1}+W_{1,2}+W_{2,1}+W_{2,3}+W_{3,1}+W_{3,3}$, 
and the colluding users' information $W_{1,1}, Z_{1,1}, W_{2,1}, Z_{2,1}$.

We now extend the above example to a general scheme applicable to arbitrary system parameters, while maintaining the same correctness and security properties.

\subsection{General Scheme for Arbitrary $U, V, U_{0}, V_{0}, T$}

We consider a system with $UV$ users indexed by $(u,v)\in[U]\times[V]$.
Each relay receives messages from at least $V_0$ users, and the server receives messages from at least $U_0$ relays. Moreover, at most $T$ users may collude with the server or with any relay.
Accordingly, we choose the input length as
$L = U_0 V_0 - T.$
Each User $(u,v)$ holds an input vector
$W_{u,v} = \bigl(W_{u,v}(1), \cdots, W_{u,v}(U_0 V_0 - T)\bigr)^\top
\in \mathbb{F}_q^{(U_0 V_0 - T)\times 1}.$

We first specify the correlated secret keys used in the scheme.
Consider a total of $2UV$ independent and uniformly distributed random vectors over $\mathbb{F}_q$, given by
$N_{u,v} = \bigl(N_{u,v}(1), \cdots, N_{u,v}(U_0 V_0 - T)\bigr)^\top
\in \mathbb{F}_q^{(U_0 V_0 - T)\times 1},$
and
$S_{u,v} = \bigl(S_{u,v}(1), \cdots, S_{u,v}(T)\bigr)^\top
\in \mathbb{F}_q^{T\times 1},$
for all $(u,v)\in[U]\times[V]$.
From these secret keys, we construct a collection of linearly independent linear combinations over $\mathbb{F}_q$, parameterized by a carefully designed MDS matrix. Specifically, for each $(u,v)\in[U]\times[V]$, the randomness available at User $(u,v)$ is defined as
\begin{equation}
Z_{u,v}
=
\left(
N_{u,v},
\bigl\{[\boldsymbol{Q}_{i,j}]_{u,v}\bigr\}_{(i,j)\in[U]\times[V]}
\right).
\label{eq:zz}
\end{equation}
Here,
\[
[\boldsymbol{Q}_{i,j}]_{u,v}
\triangleq
\Bigl(
N_{i,j}(1), \cdots, N_{i,j}(U_0 V_0 - T),
S_{i,j}(1), \cdots, S_{i,j}(T)
\Bigr)
\boldsymbol{\alpha}_{u,v}
\in \mathbb{F}_q,
\]
where $\boldsymbol{\alpha}_{u,v}\in\mathbb{F}_q^{U_0 V_0\times 1}$ denotes the $(V(u-1)+v)$th column of a matrix
$\boldsymbol{\alpha}\in\mathbb{F}_q^{U_0 V_0\times UV}$.

We say that a matrix $\boldsymbol{\alpha}\in\mathbb{F}_q^{U_0 V_0\times UV}$ with $U_0 V_0 < UV$ is an MDS matrix if any
$U_0 V_0\times U_0 V_0$ submatrix is nonsingular.
Furthermore, $\boldsymbol{\alpha}$ is said to be a $T$-privacy MDS matrix if the submatrix formed by its last $T$ rows is also MDS.
Such a matrix exists for sufficiently large field size $q$.

A $T$-private MDS matrix~\cite{so2022lightsecagg} guarantees that, for any subset
$\mathcal{T}\subseteq[U]\times[V]$ with $|\mathcal{T}|\le T$, the collection of linear projections
$\{[\boldsymbol{Q}_{i,j}]_{u,v}\}_{(i,j)\in[U]\times[V],\,(u,v)\in\mathcal{T}}$
is statistically independent of the masking variables
$\{N_{i,j}\}_{(i,j)\in[U]\times[V]}$, i.e.,
\begin{equation}
I\left(
\{[\boldsymbol{Q}_{i,j}]_{u,v}\}_{(i,j)\in[U]\times[V],\,(u,v)\in\mathcal{T}};
\{N_{i,j}\}_{(i,j)\in[U]\times[V]}
\right)=0.\label{Tprivacy}
\end{equation}
This follows from the fact that the submatrix formed by the last $T$ rows of
$\boldsymbol{\alpha}$ is MDS, which ensures that any set of at most $T$ such projections
depends only on the randomness vectors $\{S_{i,j}\}$ and reveals no information about
$\{N_{i,j}\}$.

To facilitate the subsequent security analysis, we next summarize several useful entropy properties of the secret keys in the following lemma.

\begin{lemma}
For any subset $\mathcal{T}\subseteq[U]\times[V]$ with $|\mathcal{T}|\le T$,
any $\mathcal{V}^{(2)}_u\subseteq \mathcal{V}^{(1)}_u \subseteq \{(u,v)\}_{v\in[V]}$
satisfying $|\mathcal{V}^{(2)}_u|\ge V_0$ for all $u\in[U]$, and any
$\mathcal{U}^{(2)}\subseteq \mathcal{U}^{(1)} \subseteq [U]$
satisfying $|\mathcal{U}^{(2)}|\ge U_0$, with $U_0V_0>T$, we have
\begin{equation}
I\left(
\{N_{i,j}\}_{(i,j)\in([U]\times[V])\setminus\mathcal{T}};
\{N_{i,j}\}_{(i,j)\in\mathcal{T}},
\{[\boldsymbol{Q}_{i,j}]_{u,v}\}_{(i,j)\in[U]\times[V],\,(u,v)\in\mathcal{T}}
\right)=0.
\label{Tprivacy1}
\end{equation}
\end{lemma}
\begin{IEEEproof}
    Since $|\mathcal{T}|\le T$, the $T$-privacy property established in
\eqref{Tprivacy} applies.
\begin{align}
    &I\left(
\{N_{i,j}\}_{(i,j)\in([U]\times[V])\setminus\mathcal{T}};
\{N_{i,j}\}_{(i,j)\in\mathcal{T}},
\{[\boldsymbol{Q}_{i,j}]_{u,v}\}_{(i,j)\in[U]\times[V],\,(u,v)\in\mathcal{T}}
\right)\\
    =&H\left(\{N_{i,j}\}_{(i,j)\in \mathcal{T}},\{[\boldsymbol{Q}_{i,j}]_{u,v}\}_{(i,j)\in[U]\times [V],(u,v)\in\mathcal{T}}\right)\notag\\
    &-H\left(\{N_{i,j}\}_{(i,j)\in \mathcal{T}},\{[\boldsymbol{Q}_{i,j}]_{u,v}\}_{(i,j)\in[U]\times [V],(u,v)\in\mathcal{T}}|\{N_{i,j}\}_{(i,j)\in ([U]\times [V])\setminus \mathcal{T}}\right)\\
    =&H\left(\{N_{i,j}\}_{(i,j)\in \mathcal{T}},\{[\boldsymbol{Q}_{i,j}]_{u,v}\}_{(i,j)\in[U]\times [V],(u,v)\in\mathcal{T}}\right)-H\left(\{N_{i,j}\}_{(i,j)\in \mathcal{T}}|\{N_{i,j}\}_{(i,j)\in ([U]\times [V])\setminus \mathcal{T}}\right)\notag\\
    &-H\left(\{[\boldsymbol{Q}_{i,j}]_{u,v}\}_{(i,j)\in[U]\times [V],(u,v)\in\mathcal{T}}|\{N_{i,j}\}_{(i,j)\in [U]\times [V]}\right)\\
    =&H(\{N_{i,j}\}_{(i,j)\in \mathcal{T}},\{[\boldsymbol{Q}_{i,j}]_{u,v}\bigr\}_{(i,j)\in[U]\times [V],(u,v)\in\mathcal{T}})-H(\{N_{i,j}\}_{(i,j)\in \mathcal{T}})\notag\\
    &-H(\{[\boldsymbol{Q}_{i,j}]_{u,v}\bigr\}_{(i,j)\in[U]\times [V],(u,v)\in\mathcal{T}})+\underbrace{I(\{[\boldsymbol{Q}_{i,j}]_{u,v}\bigr\}_{(i,j)\in[U]\times [V],(u,v)\in\mathcal{T}};\{N_{i,j}\}_{(i,j)\in [U]\times [V]})}_{\overset{(\ref{Tprivacy})}{=}0}\label{lemma1t1}\\
    \leq &H(\{N_{i,j}\}_{(i,j)\in \mathcal{T}})+H(\{[\boldsymbol{Q}_{i,j}]_{u,v}\bigr\}_{(i,j)\in[U]\times [V],(u,v)\in\mathcal{T}})\notag\\
    &-H(\{N_{i,j}\}_{(i,j)\in \mathcal{T}})-H(\{[\boldsymbol{Q}_{i,j}]_{u,v}\bigr\}_{(i,j)\in[U]\times [V],(u,v)\in\mathcal{T}})\\
    =&0,
\end{align}
where the fourth term in~\eqref{lemma1t1} equals zero by
\eqref{Tprivacy}, which follows from the $T$-privacy property
of the MDS matrix $\boldsymbol{\alpha}$.
\end{IEEEproof}

With the correlated secret keys fully specified, 
we proceed to describe the message transmissions over two rounds.

\textbf{First round, first hop:}
Each user transmits
\begin{align}
X^{(1)}_{u,v} = W_{u,v} + N_{u,v}, \quad \forall (u,v)\in[U]\times[V].\label{r1h1}
\end{align}

\textbf{First round, second hop:}
For any $\mathcal{V}^{(1)}_u\subseteq\{(u,v)\}_{v\in[V]}$ with
$|\mathcal{V}^{(1)}_u|\ge V_{0}=2$, 
Relay $u$ computes $Y_u^{(1)}$ and forwards it to the server.
\begin{align}
Y_u^{(1)} = \sum_{(u,v)\in\mathcal{V}^{(1)}_u} \big(W_{u,v}+N_{u,v}\big), \quad \forall u\in[U].\label{r1h2}
\end{align}

After the first round, each relay reports to the server the set of its surviving users $\mathcal{V}^{(1)}_u$.
Based on the set of surviving relays $\mathcal{U}^{(1)}$ and the reported user sets
$\{\mathcal{V}^{(1)}_u\}_{u\in\mathcal{U}^{(1)}}$,
the server determines the complete set of users that survive the first round, denoted by $\mathcal{S}^{(1)} \triangleq \cup_{u\in\mathcal{U}^{(1)}} \mathcal{V}^{(1)}_u .$
The server then broadcasts $\mathcal{S}^{(1)}$ to all surviving relays, which forward it to their associated surviving users and request the transmission of the second-round messages.

\textbf{Second round, first hop:}
To enable the recovery of the aggregated masking variables at the server,
each surviving User $(u,v)\in\mathcal{S}^{(1)}$ transmits
\begin{align}
X^{(2)}_{u,v} = \sum_{(i,j)\in\mathcal{S}^{(1)}} [\boldsymbol{Q}_{i,j}]_{u,v}. \label{r2h1}
\end{align}

\textbf{Second round, second hop:}
Since each Relay $u$ has at least $V_0$ surviving users in the second round,
i.e., $|\mathcal{V}^{(2)}_u|\ge V_0$, each Relay $u$ arbitrarily selects a subset
$\mathcal{Q}_u\subseteq\mathcal{V}^{(2)}_u$ with $|\mathcal{Q}_u|=V_{0}$,
and forwards
\begin{align}
Y^{(2)}_u=\{X^{(2)}_{u,v}\}_{v\in\mathcal{Q}_u}, \quad u\in[U]. \label{r2h2}
\end{align}

With the correlated randomness and two-round message transmissions fully specified, 
we proceed to analyze the scheme's achievable rate, correctness, and information-theoretic security.

\emph{Rate:} 
We now specify the communication rates of the proposed scheme. 
Each user input has length $L = U_0 V_0 - T$ symbols over $\mathbb{F}_q$. 

\textbf{First round rates:} 
Each user transmits 
$L_X^{(1)} = U_0 V_0 - T$ symbols, and each relay forwards 
$L_Y^{(1)} = U_0 V_0 - T$ symbols. 
Hence, the first-round rates are
\[
R_X^{(1)} = \frac{L_X^{(1)}}{L} = 1, \quad 
R_Y^{(1)} = \frac{L_Y^{(1)}}{L} = 1.
\]

\textbf{Second round rates:} 
Each user transmits a single symbol, $L_X^{(2)} = 1$, and each relay forwards 
$L_Y^{(2)} = V_0$ symbols. 
Thus, the second-round rates are
\[
R_X^{(2)} = \frac{L_X^{(2)}}{L} = \frac{1}{U_0 V_0 - T}, \quad
R_Y^{(2)} = \frac{L_Y^{(2)}}{L} = \frac{V_0}{U_0 V_0 - T} = \frac{1}{U_0 - \frac{T}{V_0}}.
\]

With these rates, the proposed scheme meets both the correctness and security requirements, as we detail in the following analysis.

\emph{Correctness:}  
The server receives the second-round messages $\{Y^{(2)}_u\}_{u\in \mathcal{U}^{(2)}}$, which by (\ref{r2h2}) satisfy $\{Y^{(2)}_u\}_{u\in \mathcal{U}^{(2)}} = \{\{X^{(2)}_{u,v}\}_{v\in\mathcal{Q}_u}\}_{u\in \mathcal{U}^{(2)}}$, and by (\ref{r2h1}) further equal $\{\{\sum_{(i,j)\in\mathcal{S}^{(1)}} [\boldsymbol{Q}_{i,j}]_{u,v}\}_{v\in\mathcal{Q}_u}\}_{u\in \mathcal{U}^{(2)}}$. Each encoded symbol $[\boldsymbol{Q}_{i,j}]_{u,v}$ is defined as $[\boldsymbol{Q}_{i,j}]_{u,v} \triangleq (N_{i,j}(1), \ldots, N_{i,j}(U_0 V_0 - T), S_{i,j}(1), \ldots, S_{i,j}(T)) \cdot \boldsymbol{\alpha}_{u,v}$, where $\boldsymbol{\alpha} = [\boldsymbol{\alpha}_{u,v}]_{(u,v)\in [U]\times [V]} \in \mathbb{F}_q^{U_0 V_0 \times U_0 V_0}$ is an MDS matrix. By the MDS property, any $U_0 V_0$ columns of $\boldsymbol{\alpha}$ are linearly independent, ensuring that the server can recover all random symbols and thus the desired sum $\sum_{(u,v)\in[U_0]\times[V_0]} W_{u,v}$.

By design, each Relay $u\in \mathcal{U}^{(2)}$ forwards messages from at least $|\mathcal{Q}_u|\ge V_0$ users, and the server receives messages from at least $|\mathcal{U}^{(2)}|\ge U_0$ relays. 
Therefore, the server obtains at least $U_0 V_0$ coded symbols of the form $\sum_{(i,j)\in\mathcal{S}^{(1)}} [\boldsymbol{Q}_{i,j}]_{u,v}$.
Since these correspond to $U_0 V_0$ linearly independent columns of the MDS matrix $\boldsymbol{\alpha}$, the server can uniquely recover the aggregated randomness vector
$\sum_{(i,j)\in \mathcal{S}^{(1)}}
(N_{i,j}(1), \ldots, N_{i,j}(U_0 V_0 - T), S_{i,j}(1), \ldots, S_{i,j}(T))$
without error. Equivalently, the server recovers
$\sum_{(i,j)\in \mathcal{S}^{(1)}} N_{i,j}$ and $\sum_{(i,j)\in \mathcal{S}^{(1)}} S_{i,j}$
exactly.

Finally, combining the recovered aggregated keys with the sum of the first-round messages,
$\sum_{u\in \mathcal{U}^{(1)}} Y^{(1)}_u
\overset{(\ref{r1h2})}{=}
\sum_{(u,v)\in \mathcal{S}^{(1)}} (W_{u,v} + N_{u,v}),$
the server can subtract the aggregated randomness $\sum_{(u,v)\in \mathcal{S}^{(1)}} N_{u,v}$ decoded in the second round, and hence uniquely recover the desired aggregation
$\sum_{(u,v)\in \mathcal{S}^{(1)}} W_{u,v}$ with zero decoding error.

{\em Security:} 
Having established the correctness of the aggregation scheme, we now analyze its security.
The system ensures that no relay or server can learn any user's input beyond what is allowed, even if up to 
$T$ users collude with them.
The security analysis consists of two parts. We first consider relay security.

\medskip
\noindent\textbf{Relay security:} 
Each relay only observes messages from its associated users. We show that even if a relay
colludes with any set of at most $T$ users, it cannot obtain any information about the inputs
of non-colluding users.
Recall that $\mathcal{M}_u$ denotes the set of users associated with relay~$u$, and
$\mathcal{T}$ denotes the set of colluding users with $|\mathcal{T}|\le T$.
We focus on the case where $|\mathcal{V}^{(1)}_u|\geq U_{0}V_{0}$, for which
\begin{align}
    &I\Big(\left\{X^{(1)}_{u,v}\right\}_{(u,v)\in \mathcal{M}_u},\left\{X^{(2)}_{u,v}\right\}_{(u,v)\in \mathcal{V}^{(1)}_{u}}; W_{[U]\times [V]}\Big|\{W_{u,v},Z_{u,v}\}_{(u,v)\in \mathcal{T}}\Big)\\
    =&H\Big(\left\{X^{(1)}_{u,v}\right\}_{(u,v)\in \mathcal{M}_u},\left\{X^{(2)}_{u,v}\right\}_{(u,v)\in \mathcal{V}^{(1)}_{u}}\Big|\{W_{u,v},Z_{u,v}\}_{(u,v)\in \mathcal{T}}\Big)\notag\\
    &-H\Big(\left\{X^{(1)}_{u,v}\right\}_{(u,v)\in \mathcal{M}_u},\left\{X^{(2)}_{u,v}\right\}_{(u,v)\in \mathcal{V}^{(1)}_{u}}\Big| W_{[U]\times [V]},\{W_{u,v},Z_{u,v}\}_{(u,v)\in \mathcal{T}}\Big)\\
    =&H\Bigg(\{W_{u,v} + N_{u,v}\}_{(u,v)\in \mathcal{M}_u},\Bigg\{\sum_{(i,j)\in\mathcal{S}^{(1)}} [\boldsymbol{Q}_{i,j}]_{u,v}\Bigg\}_{(u,v)\in \mathcal{V}^{(1)}_{u}}\Bigg|\notag\\
    &\Big\{W_{u,v},N_{u,v},\bigl\{ [\boldsymbol{Q}_{i,j}]_{u,v} \bigr\}_{(i,j)\in[U]\times[V]}\Big\}_{(u,v)\in \mathcal{T}}\Bigg)\notag\\
    &-H\Bigg(\{N_{u,v}\}_{(u,v)\in \mathcal{M}_u},\Bigg\{\sum_{(i,j)\in\mathcal{S}^{(1)}} [\boldsymbol{Q}_{i,j}]_{u,v}\Bigg\}_{(u,v)\in \mathcal{V}^{(1)}_{u}}\Bigg|\notag\\
    & W_{[U]\times [V]},\Big\{N_{u,v},\bigl\{ [\boldsymbol{Q}_{i,j}]_{u,v} \bigr\}_{(i,j)\in[U]\times[V]}\Big\}_{(u,v)\in \mathcal{T}}\Bigg)\\
    =&H\Big(\{W_{u,v} + N_{u,v}\}_{(u,v)\in \mathcal{M}_u\setminus \mathcal{T}}\Big|\Big\{W_{u,v},N_{u,v},\bigl\{ [\boldsymbol{Q}_{i,j}]_{u,v} \bigr\}_{(i,j)\in[U]\times[V]}\Big\}_{(u,v)\in \mathcal{T}}\Big)\notag\\
    &+H\Bigg(\sum_{(i,j)\in \mathcal{S}^{(1)}} N_{i,j}, \sum_{(i,j)\in \mathcal{S}^{(1)}} S_{i,j}\Bigg|\{W_{u,v} + N_{u,v}\}_{(u,v)\in \mathcal{M}_u\setminus \mathcal{T}},\notag\\
    &\Big\{W_{u,v},N_{u,v},\bigl\{ [\boldsymbol{Q}_{i,j}]_{u,v} \bigr\}_{(i,j)\in[U]\times[V]}\Big\}_{(u,v)\in \mathcal{T}}\Bigg) \notag\\
    &-H\Big(\{N_{u,v}\}_{(u,v)\in \mathcal{M}_u\setminus \mathcal{T}}\Big| W_{[U]\times [V]},\Big\{N_{u,v},\bigl\{ [\boldsymbol{Q}_{i,j}]_{u,v} \bigr\}_{(i,j)\in[U]\times[V]}\Big\}_{(u,v)\in \mathcal{T}}\Big)-\notag\\
    &H\Bigg(\sum_{(i,j)\in \mathcal{S}^{(1)}} N_{i,j}, \sum_{(i,j)\in \mathcal{S}^{(1)}} S_{i,j}\Bigg|\{ N_{u,v}\}_{(u,v)\in \mathcal{M}_u\cup\mathcal{T}}, W_{[U]\times [V]},\Big\{\bigl\{ [\boldsymbol{Q}_{i,j}]_{u,v} \bigr\}_{(i,j)\in[U]\times[V]}\Big\}_{(u,v)\in \mathcal{T}}\Bigg)\label{relaysecpft1}\\
    =&H\Big(\{W_{u,v} + N_{u,v}\}_{(u,v)\in \mathcal{M}_u\setminus \mathcal{T}}\Big|\Big\{W_{u,v},N_{u,v},\bigl\{ [\boldsymbol{Q}_{i,j}]_{u,v} \bigr\}_{(i,j)\in[U]\times[V]}\Big\}_{(u,v)\in \mathcal{T}}\Big)\notag\\
    &+H\Bigg(\sum_{(i,j)\in \mathcal{S}^{(1)}} N_{i,j}\Bigg|\{W_{u,v} + N_{u,v}\}_{(u,v)\in \mathcal{M}_u\setminus \mathcal{T}},\Big\{W_{u,v},N_{u,v},\bigl\{ [\boldsymbol{Q}_{i,j}]_{u,v} \bigr\}_{(i,j)\in[U]\times[V]}\Big\}_{(u,v)\in \mathcal{T}}\Bigg) \notag\\
    &+H\Bigg(\sum_{(i,j)\in \mathcal{S}^{(1)}} S_{i,j}\Bigg|\sum_{(i,j)\in \mathcal{S}^{(1)}} N_{i,j}, \{W_{u,v} + N_{u,v}\}_{(u,v)\in \mathcal{M}_u\setminus \mathcal{T}},\notag\\
    &\Big\{W_{u,v},N_{u,v},\bigl\{ [\boldsymbol{Q}_{i,j}]_{u,v} \bigr\}_{(i,j)\in[U]\times[V]}\Big\}_{(u,v)\in \mathcal{T}}\Bigg) \notag\\
    &-H\Big(\{N_{u,v}\}_{(u,v)\in \mathcal{M}_u\setminus \mathcal{T}}\Big| W_{[U]\times [V]},\Big\{N_{u,v},\bigl\{ [\boldsymbol{Q}_{i,j}]_{u,v} \bigr\}_{(i,j)\in[U]\times[V]}\Big\}_{(u,v)\in \mathcal{T}}\Big)\notag\\
    &-H\Bigg(\sum_{(i,j)\in \mathcal{S}^{(1)}} N_{i,j}\Bigg|\{ N_{u,v}\}_{(u,v)\in \mathcal{M}_u\cup\mathcal{T}}, W_{[U]\times [V]},\Big\{\bigl\{ [\boldsymbol{Q}_{i,j}]_{u,v} \bigr\}_{(i,j)\in[U]\times[V]}\Big\}_{(u,v)\in \mathcal{T}}\Bigg)-\notag\\
    &H\Bigg(\sum_{(i,j)\in \mathcal{S}^{(1)}} S_{i,j}\Bigg| \sum_{(i,j)\in \mathcal{S}^{(1)}} N_{i,j},\{ N_{u,v}\}_{(u,v)\in \mathcal{M}_u\cup\mathcal{T}}, W_{[U]\times [V]},\Big\{\bigl\{ [\boldsymbol{Q}_{i,j}]_{u,v} \bigr\}_{(i,j)\in[U]\times[V]}\Big\}_{(u,v)\in \mathcal{T}}\Bigg)\\
    \overset{(\ref{independent})}{\leq}&H\Big(\{W_{u,v} + N_{u,v}\}_{(u,v)\in \mathcal{M}_u\setminus \mathcal{T}}\Big)+H\Bigg(\sum_{(i,j)\in \mathcal{S}^{(1)}} N_{i,j}\Bigg)\notag\\
    & +\underbrace{H\Bigg(\sum_{(i,j)\in \mathcal{S}^{(1)}} S_{i,j}\Bigg|\sum_{(i,j)\in \mathcal{S}^{(1)}} N_{i,j}, \bigl\{ [\boldsymbol{Q}_{i,j}]_{u,v} \bigr\}_{(i,j)\in[U]\times[V],(u,v)\in \mathcal{T}}\Bigg)}_{\overset{(\ref{eq:zz})}{=}0} -H\Big(\{N_{u,v}\}_{(u,v)\in \mathcal{M}_u\setminus \mathcal{T}}\Big)\notag\\
    &-\underbrace{I\Bigg(\{N_{u,v}\}_{(u,v)\in \mathcal{M}_u\setminus \mathcal{T}}; \Big\{N_{u,v},\bigl\{ [\boldsymbol{Q}_{i,j}]_{u,v} \bigr\}_{(i,j)\in[U]\times[V]}\Big\}_{(u,v)\in \mathcal{T}}\Bigg)}_{\overset{(\ref{Tprivacy1})}{=}0}\notag\\
    &-H\Bigg(\sum_{(i,j)\in \mathcal{S}^{(1)}} N_{i,j}\Bigg|\{ N_{u,v}\}_{(u,v)\in \mathcal{M}_u\cup\mathcal{T}},\Big\{\bigl\{ [\boldsymbol{Q}_{i,j}]_{u,v} \bigr\}_{(i,j)\in[U]\times[V]}\Big\}_{(u,v)\in \mathcal{T}}\Bigg)\label{relaysecpft2}\\
    \leq & |\mathcal{M}_u\setminus \mathcal{T}|L+L-|\mathcal{M}_u\setminus \mathcal{T}|L-L=0,\label{relaysecpft3}
\end{align}
where (\ref{relaysecpft1}) holds because $|\mathcal{V}^{(1)}_u|\geq U_0V_0$.
Hence, from the collection
$\bigl\{\sum_{(i,j)\in\mathcal{S}^{(1)}} [\boldsymbol{Q}_{i,j}]_{u,v}\bigr\}_{(u,v)\in\mathcal{V}^{(1)}_u}$,
relay~$u$ can decode the aggregated randomness
$\sum_{(i,j)\in\mathcal{S}^{(1)}} N_{i,j}$ and
$\sum_{(i,j)\in\mathcal{S}^{(1)}} S_{i,j}$ due to the MDS property of the encoding matrix.
In (\ref{relaysecpft2}), the third entropy term
$H\Big(\sum_{(i,j)\in \mathcal{S}^{(1)}} S_{i,j}\,\Big|\,
\sum_{(i,j)\in \mathcal{S}^{(1)}} N_{i,j},
\bigl\{[\boldsymbol{Q}_{i,j}]_{u,v}\bigr\}_{(i,j)\in[U]\times[V],(u,v)\in\mathcal{T}}
\Big)$
equals zero, since $\sum_{(i,j)\in \mathcal{S}^{(1)}} S_{i,j}$ is uniquely determined by
$\sum_{(i,j)\in \mathcal{S}^{(1)}} N_{i,j}$ together with the colluding users' coded symbols,
as ensured by the key construction in~(\ref{eq:zz}).
Moreover, the mutual information term
\be 
I\Big(\{N_{u,v}\}_{(u,v)\in \mathcal{M}_u\setminus \mathcal{T}};
\{N_{u,v},[\boldsymbol{Q}_{i,j}]_{u,v}\}_{(u,v)\in \mathcal{T}}\Big) 
\ee
is zero due to the $T$-privacy property in~(\ref{Tprivacy1}).
Finally, the last conditional entropy term equals $L$ since
$\sum_{(i,j)\in \mathcal{S}^{(1)}} N_{i,j}$ is independent of
$\{ N_{u,v}\}_{(u,v)\in \mathcal{M}_u\cup\mathcal{T}}$ and
$\bigl\{[\boldsymbol{Q}_{i,j}]_{u,v}\bigr\}_{(i,j)\in[U]\times[V],(u,v)\in \mathcal{T}}$.

When $|\mathcal{V}^{(1)}_u|< U_{0}V_{0}$, we have
\begin{align}
    &I\Big(\left\{X^{(1)}_{u,v}\right\}_{(u,v)\in \mathcal{M}_u},\left\{X^{(2)}_{u,v}\right\}_{(u,v)\in \mathcal{V}^{(1)}_{u}}; W_{[U]\times [V]}\Big|\{W_{u,v},Z_{u,v}\}_{(u,v)\in \mathcal{T}}\Big)\\
    \leq&I\Big(\left\{X^{(1)}_{u,v}\right\}_{(u,v)\in \mathcal{M}_u},\{X^{(2)}_{u,v}\}_{(u,v)\in [U_0]\times [V_0]}; W_{[U]\times [V]}\Big|\{W_{u,v},Z_{u,v}\}_{(u,v)\in \mathcal{T}}\Big)\overset{(\ref{relaysecpft3})}{=}0.
\end{align}
Hence, relay~$u$ obtains no information about the users' messages, and the relay
security constraint is satisfied.

We next analyze the server security of the proposed scheme.

\noindent\textbf{Server security:}
The aggregation server collects messages from multiple relays.
The proposed scheme guarantees that, even if the server colludes with any set of at most $T$ users,
it can only learn the aggregation of the surviving users in the first round,
and obtains no additional information about the messages of the non-colluding users.
\begin{align}
&I\Bigg( \left\{Y^{(1)}_u\right\}_{u\in [U]}, \left\{Y^{(2)}_u\right\}_{u\in \mathcal{U}^{(1)}}; W_{[U]\times [V]} \Bigg|\sum_{(u,v)\in \mathcal{S}^{(1)}} W_{u,v}, \{W_{u,v}, Z_{u,v}\}_{(u,v)\in \mathcal{T}} \Bigg)\\
=&H\Bigg( \left\{Y^{(1)}_u\right\}_{u\in [U]}, \left\{Y^{(2)}_u\right\}_{u\in \mathcal{U}^{(1)}} \Bigg|\sum_{(u,v)\in \mathcal{S}^{(1)}} W_{u,v}, \{W_{u,v}, Z_{u,v}\}_{(u,v)\in \mathcal{T}} \Bigg)\notag\\
&-H\Bigg( \left\{Y^{(1)}_u\right\}_{u\in [U]}, \left\{Y^{(2)}_u\right\}_{u\in \mathcal{U}^{(1)}}\Bigg| W_{[U]\times [V]} ,\sum_{(u,v)\in \mathcal{S}^{(1)}} W_{u,v}, \{W_{u,v}, Z_{u,v}\}_{(u,v)\in \mathcal{T}} \Bigg)\\
=&H\Bigg( \Bigg\{\sum_{(u,v)\in\mathcal{V}^{(1)}_u} \big(W_{u,v}+N_{u,v}\big)\Bigg\}_{u\in [U]}, \Bigg\{\sum_{(i,j)\in\mathcal{S}^{(1)}} [\boldsymbol{Q}_{i,j}]_{u,v}\Bigg\}_{(u,v)\in \mathcal{S}^{(1)}} \Bigg|\sum_{(u,v)\in \mathcal{S}^{(1)}} W_{u,v},\notag\\
&\Big\{W_{u,v},N_{u,v},\bigl\{ [\boldsymbol{Q}_{i,j}]_{u,v} \bigr\}_{(i,j)\in[U]\times[V]}\Big\}_{(u,v)\in \mathcal{T}}\Bigg)-H\Bigg( \Bigg\{\sum_{(u,v)\in\mathcal{V}^{(1)}_u} \big(N_{u,v}\big)\Bigg\}_{u\in [U]}, \notag\\
&\Bigg\{\sum_{(i,j)\in\mathcal{S}^{(1)}} [\boldsymbol{Q}_{i,j}]_{u,v}\Bigg\}_{(u,v)\in \mathcal{S}^{(1)}} \Bigg|W_{[U]\times [V]},\Big\{N_{u,v},\bigl\{ [\boldsymbol{Q}_{i,j}]_{u,v} \bigr\}_{(i,j)\in[U]\times[V]}\Big\}_{(u,v)\in \mathcal{T}}\Big)\\
=&H\Bigg( \Bigg\{\sum_{(u,v)\in\mathcal{V}^{(1)}_u} \big(W_{u,v}+N_{u,v}\big)\Bigg\}_{u\in [U]\setminus \mathcal{U}_T} \Bigg|\sum_{(u,v)\in \mathcal{S}^{(1)}} W_{u,v},\notag\\
&\Big\{W_{u,v},N_{u,v},\bigl\{ [\boldsymbol{Q}_{i,j}]_{u,v} \bigr\}_{(i,j)\in[U]\times[V]}\Big\}_{(u,v)\in \mathcal{T}}\Bigg)+H\Bigg(  \sum_{(i,j)\in \mathcal{S}^{(1)}} N_{i,j}, \sum_{(i,j)\in \mathcal{S}^{(1)}} S_{i,j} \Bigg|\notag\\
&\Bigg\{\sum_{(u,v)\in\mathcal{V}^{(1)}_u} \big(W_{u,v}+N_{u,v}\big)\Bigg\}_{u\in [U]},\sum_{(u,v)\in \mathcal{S}^{(1)}} W_{u,v},\Big\{W_{u,v},N_{u,v},\bigl\{ [\boldsymbol{Q}_{i,j}]_{u,v} \bigr\}_{(i,j)\in[U]\times[V]}\Big\}_{(u,v)\in \mathcal{T}}\Bigg)\notag\\
&-H\Bigg( \Bigg\{\sum_{(u,v)\in\mathcal{V}^{(1)}_u} \big(N_{u,v}\big)\Bigg\}_{u\in [U]\setminus \mathcal{U}_T} \Bigg|\Big\{N_{u,v},\bigl\{ [\boldsymbol{Q}_{i,j}]_{u,v} \bigr\}_{(i,j)\in[U]\times[V]}\Big\}_{(u,v)\in \mathcal{T}},W_{[U]\times [V]}\Bigg)\notag\\
&-H\Bigg(  \sum_{(i,j)\in \mathcal{S}^{(1)}} N_{i,j}, \sum_{(i,j)\in \mathcal{S}^{(1)}} S_{i,j} \Bigg|\Bigg\{\sum_{(u,v)\in\mathcal{V}^{(1)}_u} \big(W_{u,v}+N_{u,v}\big)\Bigg\}_{u\in [U]},W_{[U]\times [V]},\notag\\
&\Big\{N_{u,v},\bigl\{ [\boldsymbol{Q}_{i,j}]_{u,v} \bigr\}_{(i,j)\in[U]\times[V]}\Big\}_{(u,v)\in \mathcal{T}}\Bigg)\label{serversecpft1}\\
=&H\Bigg( \Bigg\{\sum_{(u,v)\in\mathcal{V}^{(1)}_u} \big(W_{u,v}+N_{u,v}\big)\Bigg\}_{u\in [U]\setminus \mathcal{U}_T} \Bigg|\Big\{W_{u,v},N_{u,v},\bigl\{ [\boldsymbol{Q}_{i,j}]_{u,v} \bigr\}_{(i,j)\in[U]\times[V]}\Big\}_{(u,v)\in \mathcal{T}}\Bigg)\notag\\
&+H\Bigg(  \sum_{(i,j)\in \mathcal{S}^{(1)}} N_{i,j} \Bigg|\Bigg\{\sum_{(u,v)\in\mathcal{V}^{(1)}_u} \big(W_{u,v}+N_{u,v}\big)\Bigg\}_{u\in [U]},\sum_{(u,v)\in \mathcal{S}^{(1)}} W_{u,v},\notag\\
&\Big\{W_{u,v},N_{u,v},\bigl\{ [\boldsymbol{Q}_{i,j}]_{u,v} \bigr\}_{(i,j)\in[U]\times[V]}\Big\}_{(u,v)\in \mathcal{T}}\Bigg)+H\Bigg(  \sum_{(i,j)\in \mathcal{S}^{(1)}} S_{i,j} \Bigg| \sum_{(i,j)\in \mathcal{S}^{(1)}} N_{i,j}, \notag\\
&\Bigg\{\sum_{(u,v)\in\mathcal{V}^{(1)}_u} \big(W_{u,v}+N_{u,v}\big)\Bigg\}_{u\in [U]},\sum_{(u,v)\in \mathcal{S}^{(1)}} W_{u,v},\Big\{W_{u,v},N_{u,v},\bigl\{ [\boldsymbol{Q}_{i,j}]_{u,v} \bigr\}_{(i,j)\in[U]\times[V]}\Big\}_{(u,v)\in \mathcal{T}}\Bigg)\notag\\
&-H\Bigg( \Bigg\{\sum_{(u,v)\in\mathcal{V}^{(1)}_u} \big(N_{u,v}\big)\Bigg\}_{u\in [U]\setminus \mathcal{U}_T}\Bigg)-H\Bigg(  \sum_{(i,j)\in \mathcal{S}^{(1)}} N_{i,j} \Bigg|\Bigg\{\sum_{(u,v)\in\mathcal{V}^{(1)}_u} N_{u,v}\Bigg\}_{u\in [U]},\notag\\
&W_{[U]\times [V]},\Big\{N_{u,v},\bigl\{ [\boldsymbol{Q}_{i,j}]_{u,v} \bigr\}_{(i,j)\in[U]\times[V]}\Big\}_{(u,v)\in \mathcal{T}}\Bigg)-H\Bigg(  \sum_{(i,j)\in \mathcal{S}^{(1)}} S_{i,j} \Bigg| \sum_{(i,j)\in \mathcal{S}^{(1)}} N_{i,j},\notag\\
&\Bigg\{\sum_{(u,v)\in\mathcal{V}^{(1)}_u} N_{u,v}\Bigg\}_{u\in [U]}, W_{[U]\times [V]},\Big\{N_{u,v},\bigl\{ [\boldsymbol{Q}_{i,j}]_{u,v} \bigr\}_{(i,j)\in[U]\times[V]}\Big\}_{(u,v)\in \mathcal{T}}\Bigg)\label{serversecpft2}\\
\leq&H\Bigg( \Bigg\{\sum_{(u,v)\in\mathcal{V}^{(1)}_u} \big(W_{u,v}+N_{u,v}\big)\Bigg\}_{u\in [U]\setminus \mathcal{U}_T} \Bigg)-H\Bigg( \Bigg\{\sum_{(u,v)\in\mathcal{V}^{(1)}_u} N_{u,v}\Bigg\}_{u\in [U]\setminus \mathcal{U}_T} \Bigg)\label{serversecpft3}\\
=&|[U]\setminus \mathcal{U}_T|- |[U]\setminus \mathcal{U}_T|=0.
\end{align}
We next justify the intermediate steps (\ref{serversecpft1})--(\ref{serversecpft2}).

First, let $\mathcal{U}_T$ denote the set of relays that are connected exclusively to the colluding users in $\mathcal{T}$.  
Conditioned on the values $\{W_{u,v}, Z_{u,v}\}_{(u,v)\in\mathcal{T}}$, the messages transmitted by these relays in $\mathcal{U}_T$ are completely determined, and thus they introduce no additional uncertainty.
In (\ref{serversecpft1}), the second and the fourth entropy terms follow from the MDS structure of the second-round encoding.
Specifically, the collection $\{Y^{(2)}_u\}_{u\in\mathcal{U}^{(1)}}$ allows the server to recover only the aggregated
randomness $\sum_{(i,j)\in\mathcal{S}^{(1)}} N_{i,j}$ and $\sum_{(i,j)\in\mathcal{S}^{(1)}} S_{i,j}$.
Given the aggregate $\sum_{(u,v)\in\mathcal{S}^{(1)}} W_{u,v}$, these quantities are independent of the individual users'
messages and thus can be separated as shown.
In (\ref{serversecpft2}), the second term is zero since
$\sum_{(i,j)\in \mathcal{S}^{(1)}} N_{i,j}$ is uniquely determined by
$\{\sum_{(u,v)\in\mathcal{V}^{(1)}_u} (W_{u,v}+N_{u,v})\}_{u\in[U]}$
together with $\sum_{(u,v)\in\mathcal{S}^{(1)}} W_{u,v}$.
The fourth term follows from the $T$-privacy property in~(\ref{Tprivacy1}), which guarantees that
$\{\sum_{(u,v)\in\mathcal{V}^{(1)}_u} N_{u,v}\}_{u\in[U]\setminus\mathcal{U}_T}$
is independent of the information available to the colluding users.
The third term equals the sixth term.
When $|\mathcal{T}| = T$, both terms are zero since
$\sum_{(i,j)\in \mathcal{S}^{(1)}} S_{i,j}$ is uniquely determined by
$\sum_{(i,j)\in \mathcal{S}^{(1)}} N_{i,j}$ and the second-round encoding coefficients
known to the colluding users.
When $|\mathcal{T}| < T$, the random variable
$\sum_{(i,j)\in \mathcal{S}^{(1)}} S_{i,j}$ is independent of
$\sum_{(i,j)\in \mathcal{S}^{(1)}} N_{i,j}$,
$\{\sum_{(u,v)\in\mathcal{V}^{(1)}_u} (W_{u,v}+N_{u,v})\}_{u\in [U]}$,
$\sum_{(u,v)\in \mathcal{S}^{(1)}} W_{u,v}$,
and $\{W_{u,v}, N_{u,v}, \{[\boldsymbol{Q}_{i,j}]_{u,v}\}_{(i,j)\in[U]\times[V]}\}_{(u,v)\in \mathcal{T}}$.
Hence, the third term equals
$H\!\left(\sum_{(i,j)\in \mathcal{S}^{(1)}} S_{i,j}\right)$.
Similarly, in the sixth term,
$\sum_{(i,j)\in \mathcal{S}^{(1)}} S_{i,j}$ is independent of
$\sum_{(i,j)\in \mathcal{S}^{(1)}} N_{i,j}$,
$\{\sum_{(u,v)\in\mathcal{V}^{(1)}_u} N_{u,v}\}_{u\in [U]}$,
$W_{[U]\times [V]}$,
and 
$$\{N_{u,v}, \{[\boldsymbol{Q}_{i,j}]_{u,v}\}_{(i,j)\in[U]\times[V]}\}_{(u,v)\in \mathcal{T}}$$.
Therefore, the sixth term also equals
$H\!\left(\sum_{(i,j)\in \mathcal{S}^{(1)}} S_{i,j}\right)$.

\section{Converse Proof of Theorem \ref{thm:main}}\label{sec:con}
Before presenting the converse proof, we first establish a basic property that follows from the independence between the inputs
$\{W_{u,v}\}_{(u,v)\in[U]\times [V]}$
and the secret keys
$\{Z_{u,v}\}_{(u,v)\in[U]\times [V]}$,
together with the uniform distribution of
$\{W_{u,v}\}_{(u,v)\in[U]\times [V]}$.
This property is formalized in the following lemma, which will be repeatedly invoked in the subsequent analysis.

\begin{lemma}\label{lemma:indc}
For any $V_1\leq V_2\leq V_3\leq V$, and $U_1<U_2<U_3<U$, the following equality holds. 
\begin{eqnarray}
I\Bigg( \sum_{(u,v)\in[U_2]\times [V_2]} W_{u,v}; \sum_{(u,v)\in[U_3]\times [V_3]} W_{u,v}, \left\{ W_{u,v}, Z_{u,v} \right\}_{(u,v)\in[U_1]\times [V_1]} \Bigg) = 0. \label{eq:indc}
\end{eqnarray}
\end{lemma}

{\it Proof:}
\begin{align}
& I\Bigg( \sum_{(u,v)\in[U_2]\times [V_2]} W_{u,v}; \sum_{(u,v)\in[U_3]\times [V_3]} W_{u,v}, \left\{ W_{u,v}, Z_{u,v} \right\}_{(u,v)\in[U_1]\times [V_1]} \Bigg) \notag\\
\overset{(\ref{independent})}{=}& I\Bigg( \sum_{(u,v)\in[U_2]\times [V_2]} W_{u,v}; \sum_{(u,v)\in[U_3]\times [V_3]} W_{u,v}, \left\{ W_{u,v} \right\}_{(u,v)\in[U_1]\times [V_1]} \Bigg) \\
=&  H\Bigg( \sum_{(u,v)\in[U_2]\times [V_2]} W_{u,v} \Bigg)-I\Bigg( \sum_{(u,v)\in[U_2]\times [V_2]} W_{u,v}\Bigg| \sum_{(u,v)\in[U_3]\times [V_3]} W_{u,v}, \left\{ W_{u,v} \right\}_{(u,v)\in[U_1]\times [V_1]} \Bigg) \\
=& L -H\Bigg( \sum_{(u,v)\in[U_2]\times [V_2]} W_{u,v}\Bigg| \sum_{(u,v)\in[U_3]\times [V_3]} W_{u,v}, \left\{ W_{u,v} \right\}_{(u,v)\in[U_1]\times [V_1]} \Bigg)\label{eq:c1} \\
=& L - \big( (U_1V_1 +2) L - (U_1V_1 + 1)L \big)= 0,
\end{align}
where in (\ref{eq:c1}) and the last step, we use the uniformity of $\left\{W_{u,v}\right\}_{(u,v)\in[U]\times [V]}$.

Building on Lemma~\ref{lemma:indc}, we are ready to establish the infeasible regime when $U_{0}V_{0} \leq T$ in Subsection~\ref{inf}, and the converse bounds when $U_{0}V_{0} > T$ in Subsections~\ref{RX1RY1}, \ref{RX2RY2} and~\ref{RX2RY22}.

\subsection{Infeasibility Proof When $U_{0}V_{0} \leq T$} \label{inf}

We show that when $U_{0}V_{0} \leq T$, the server may collude with all surviving users in the second round.
Due to the correctness requirement, the server can recover the sum of the inputs in the first round.
On the other hand, the server security constraint requires that the server learns nothing beyond this sum.
In particular, the server should not obtain any information about any subset of the first-round inputs.
When $U_{0}V_{0} \leq T$, the server is able to infer information about subsets of the first-round inputs, which contradicts the server security constraint.
To see why a contradiction arises, consider the following choice of sets.
Let $\mathcal{V}^{(1)}_u = \mathcal{V}^{(2)}_u = {(u,v)}_{v\in [V_{0}]}$ for each $u$, and let $\mathcal{U}^{(1)} = [U_{0}+2]$, $\mathcal{U}^{(2)} = [U_{0}]$.
Moreover, define the colluding set of users as $\mathcal{T} = [U_{0}]\times [V_{0}]$.
Note that $|\mathcal{T}| = U_{0}V_{0} \leq T$, and hence this choice of $\mathcal{V}^{(1)}$, $\mathcal{V}^{(2)}$, $\mathcal{U}^{(1)}$, and $\mathcal{T}$ is feasible under the threat model.
From the server security constraint in (\ref{serversecurity}), we have
\begin{align}
0 \overset{(\ref{serversecurity})}{=}& I\Bigg(\left\{W_{u,v}\right\}_{(u,v)\in[U]\times [V]}; \left\{Y^{(1)}_u\right\}_{u\in[U]}, \left\{Y_u^{(2)}\right\}_{u\in[U_{0}+2]} \Bigg|\notag\\
&\sum_{(u,v)\in[U_{0}+2]\times [V_{0}]} W_{u,v}, \left\{ W_{u,v}, Z_{u,v} \right\}_{(u,v)\in[U_{0}]\times [V_{0}]} \Bigg) \\
\geq&  I\Bigg( \sum_{(u,v)\in[U_{0}+1]\times [V_{0}]} W_{u,v};\left\{Y^{(1)}_u\right\}_{u\in[U_{0}+1]} \Bigg|\sum_{(u,v)\in[U_{0}+2]\times [V_{0}]} W_{u,v}, \left\{ W_{u,v}, Z_{u,v} \right\}_{(u,v)\in[U_{0}]\times [V_{0}]} \Bigg) \\
=&  I\Bigg( \sum_{(u,v)\in[U_{0}+1]\times [V_{0}]} W_{u,v};\left\{Y^{(1)}_u\right\}_{u\in[U_{0}+1]},\left\{Y_u^{[U_{0}+1]}\right\}_{u\in[U_{0}]}, \left\{X_{u,v}^{(2)}\right\}_{(u,v)\in [U_{0}]\times [V_{0}]}  \Bigg|\notag \\
&\sum_{(u,v)\in[U_{0}+2]\times [V_{0}]} W_{u,v}, \left\{ W_{u,v}, Z_{u,v} \right\}_{(u,v)\in[U_{0}]\times [V_{0}]} \Bigg) \label{infeasiblet1}\\
\geq&  I\Bigg( \sum_{(u,v)\in[U_{0}+1]\times [V_{0}]} W_{u,v};\left\{Y^{(1)}_u\right\}_{u\in[U_{0}+1]},\left\{Y_u^{[U_{0}+1]}\right\}_{u\in[U_{0}]}  ,\notag \\
&\sum_{(u,v)\in[U_{0}+2]\times [V_{0}]} W_{u,v}, \left\{ W_{u,v}, Z_{u,v} \right\}_{(u,v)\in[U_{0}]\times [V_{0}]} \Bigg) \notag\\
&-\underbrace{I\Bigg( \sum_{(u,v)\in[U_{0}+1]\times [V_{0}]} W_{u,v}; \sum_{(u,v)\in[U_{0}+2]\times [V_{0}]} W_{u,v}, \left\{ W_{u,v}, Z_{u,v} \right\}_{(u,v)\in[U_{0}]\times [V_{0}]} \Bigg)}_{\overset{(\ref{eq:indc})}{=}0} \label{inf:tt1}\\
\geq&  I\Bigg( \sum_{(u,v)\in[U_{0}+1]\times [V_{0}]} W_{u,v};\left\{Y^{(1)}_u\right\}_{u\in[U_{0}+1]},\left\{Y_u^{[U_{0}+1]}\right\}_{u\in[U_{0}]}\Bigg)   \\
=&  H\Bigg( \sum_{(u,v)\in[U_{0}+1]\times [V_{0}]} W_{u,v}\Bigg) - \underbrace{H\Bigg( \sum_{(u,v)\in[U_{0}+1]\times [V_{0}]} W_{u,v}\Bigg|\left\{Y^{(1)}_u\right\}_{u\in[U_{0}+1]},\left\{Y_u^{[U_{0}+1]}\right\}_{u\in[U_{0}]}\Bigg)}_{\overset{(\ref{correctness})}{=}0}  \label{infeasiblet2}\\
=& L,\\
\Rightarrow& 0 \geq L,
\end{align}
where $\{Y_u^{[U_{0}+1]}\}_{u\in[U_{0}]}$ denotes the second-hop messages transmitted in the second round. These messages are designed such that the server can recover the input sum
$\sum_{(u,v)\in [U_{0}+1]\times [V_{0}]} W_{u,v}$.
Moreover, by combining the first-round messages
$\{Y_u^{(1)}\}_{u\in[U_{0}+1]} \subset \{Y_u^{(1)}\}_{u\in[U_{0}+2]} \subset \{Y_u^{(1)}\}_{u\in[U]}$
with the second-round messages
$\{Y_u^{[U_{0}+1]}\}_{u\in[U_{0}]}$,
the server is able to reconstruct this sum.
In (\ref{infeasiblet1}), we use the fact that
$\{X_{u,v}^{(2)}\}_{(u,v)\in [U_{0}]\times [V_{0}]}$
is a function of
$\{W_{u,v}, Z_{u,v}\}_{(u,v)\in [U_{0}]\times [V_{0}]}$
as shown in (\ref{r2h1}), and that
$\{Y_u^{[U_{0}+1]}\}_{u\in[U_{0}]}$
is a function of
$\{X_{u,v}^{(2)}\}_{(u,v)\in [U_{0}]\times [V_{0}]}$
as shown in (\ref{r2h2}).
Here, we deliberately use the notation
$\{Y_u^{[U_{0}+1]}\}_{u\in[U_{0}]}$
instead of
$\{Y_u^{(2)}\}_{u\in[U_{0}]}$
to emphasize that these messages are specifically constructed to enable the server to recover
$\sum_{(u,v)\in [U_{0}+1]\times [V_{0}]} W_{u,v}$.
The second term in (\ref{inf:tt1}) equals zero due to Lemma~\ref{lemma:indc}, by setting
$V_1=V_2=V_3=V_{0}$,
$U_1=U_{0}$,
$U_2=U_{0}+1$,
and
$U_3=U_{0}+2$
in (\ref{eq:indc}).
In (\ref{infeasiblet2}), the first term equals $L$ since the inputs are independent and uniformly distributed, implying that their sum is also uniform.
The second term is zero due to the correctness constraint in (\ref{correctness}), when
$\mathcal{V}^{(1)}_u=\mathcal{V}^{(2)}_u=\{(u,v)\}_{v\in [V_{0}]}$,
$\mathcal{U}^{(1)}=[U_{0}+1]$,
and
$\mathcal{U}^{(2)}=[U_{0}]$.
In the final step, the inequality $0 \geq L$ yields a contradiction. Therefore, the constraints used in the above derivation cannot be satisfied simultaneously, and the problem is infeasible when $U_{0}V_{0} \leq T$.

\subsection{Converse for $R^{(1)}_X \geq 1$ and $R^{(1)}_Y \geq 1$ When $U_{0}V_{0} > T$} \label{RX1RY1}

Intuitively, the converse follows from the fact that the aggregation server must be able to recover the sum of all users that survive the first round, even if a subset of these users drop out in the second round. 
From the correctness requirement, the first-round messages must already contain sufficient information to reconstruct the inputs of those users who may potentially drop out later. 
In particular, since any user surviving the first round may be absent in the second round, the first-round communication must individually convey the input of each surviving user. 
As a result, the first-round message rate must be at least $L$ in order to account for the input of any user that drops out in the second round.

We establish the converse bound for the first hop message rate in the first round. 
Fix any $(u',v') \in [U]\times [V]$, and define
$\{\mathcal{V}^{(1)}_u\}_{u\in \mathcal{U}^{(1)}} = [U]\times [V], 
\{\mathcal{V}^{(2)}_u\}_{u\in \mathcal{U}^{(2)}} = ([U]\times [V])\setminus \{(u',v')\},$
with $\mathcal{U}^{(1)} = \mathcal{U}^{(2)} = [U]$.
From the correctness constraint in \eqref{correctness}, we have
\begin{align}
0
&\overset{(\ref{correctness})}{=} H\Bigg(\sum_{(u,v)\in[U]\times [V]} W_{u,v} \Bigg| \left\{Y^{(1)}_u\right\}_{u\in[U]}, \left\{Y^{(2)}_u\right\}_{u\in[U]} \Bigg) \\
&\geq H\Bigg(\sum_{(u,v)\in[U]\times [V]} W_{u,v} \Bigg| \left\{Y^{(1)}_u\right\}_{u\in[U]},  \left\{X^{(1)}_{u,v}\right\}_{(u,v)\in[U]\times [V]}, \left \{Y^{(2)}_u\right\}_{u\in[U]}, \notag\\
&\hspace{2.5cm}
\left\{X^{(2)}_{u,v}\right\}_{(u,v)\in([U]\times [V])\setminus\{(u',v')\}},
\left\{W_{u,v}, Z_{u,v}\right\}_{(u,v)\in([U]\times [V])\setminus\{(u',v')\}}\Bigg) \\
&\overset{(\ref{messageX2}),(\ref{messageY1}),(\ref{messageY2})}{=} 
H\Big( W_{u',v'} \,\big|\, X^{(1)}_{u',v'}, \left\{W_{u,v}, Z_{u,v}\right\}_{(u,v)\in([U]\times [V])\setminus\{(u',v')\}} \Big),\label{covpftt1}
\end{align}
where~(\ref{covpftt1}) follows since
$\{X^{(1)}_{u,v}\}_{(u,v)\in([U]\times [V])\setminus \{(u',v')\}}$ and 
$\{X^{(2)}_{u,v}\}_{(u,v)\in([U]\times [V])\setminus \{(u',v')\}}$ 
are deterministic functions of
$\{W_{u,v}, Z_{u,v}\}_{(u,v)\in([U]\times [V])\setminus \{(u',v')\}}$,
as defined in~(\ref{messageX1}) and~(\ref{messageX2}), and
$\{Y^{(1)}_u\}_{u\in[U]}$ and 
$\{Y^{(2)}_u\}_{u\in[U]}$ 
are deterministic functions of
$\{X^{(1)}_{u,v}\}_{(u,v)\in[U]\times [V]}$ and
$\{X^{(2)}_{u,v}\}_{(u,v)\in([U]\times [V])\setminus \{(u',v')\}}$, respectively,
as defined in~(\ref{messageY1}) and~(\ref{messageY2}).
Next,
\begin{eqnarray}
L &\overset{(\ref{inputsize})}{=}& H( W_{u',v'}) \overset{(\ref{independent})}{=} H\left(  W_{u',v'} \big|  \left\{ W_{u,v}, Z_{u,v} \right\}_{(u,v)\in([U]\times [V])\backslash\{(u',v')\}} \right) \\
&=& I\left(  W_{u',v'};X^{(1)}_{u',v'} \big|  \left\{ W_{u,v}, Z_{u,v} \right\}_{(u,v)\in([U]\times [V])\backslash\{(u',v')\}} \right) ]\notag\\
&&+\underbrace{H\Big( W_{u',v'} \,\big|\, X^{(1)}_{u',v'}, \left\{W_{u,v}, Z_{u,v}\right\}_{(u,v)\in([U]\times [V])\setminus\{(u',v')\}} \Big)}_{\overset{(\ref{covpftt1})}{=}0}\\
&\leq& H\left(  X^{(1)}_{u',v'} \big|  \left\{ W_{u,v}, Z_{u,v} \right\}_{(u,v)\in([U]\times [V])\backslash\{(u',v')\}} \right) \label{eq:cr1} \\ 
&\leq& H\left(X^{(1)}_{u',v'} \right) ~\leq~ L^{(1)}_X \label{eq:cr2}\\
& \Rightarrow & R_1  \overset{(\ref{rate})}{=} \frac{L^{(1)}_X}{L} \geq 1.
\end{eqnarray}

We prove the converse bound for the second hop message rate in the first round. Intuitively, the converse for the second hop follows from the fact that the server must be able to recover the aggregate contribution of each relay from the first-round messages, even if that relay drops out in the second round. 
From the correctness requirement, the information transmitted in the first round over the second hop must already suffice to reconstruct the sum of all users served by any relay that may be absent in the second round. 
Since any relay can potentially drop out after the first round, the first-round second-hop message must convey at least $L$ bits corresponding to the aggregated inputs of the users associated with that relay.

Fix any $u' \in [U]$, and define
$\{\mathcal{V}^{(1)}_u\}_{u\in \mathcal{U}^{(1)}} = [U]\times [V],  
\{\mathcal{V}^{(2)}_u\}_{u\in \mathcal{U}^{(2)}} = ([U]\setminus\{u'\})\times [V],$
with $\mathcal{U}^{(1)} = [U]$ and $\mathcal{U}^{(2)} = [U]\setminus\{u'\}$.
From the correctness constraint in~\eqref{correctness}, we have
\begin{align}
0
&\overset{(\ref{correctness})}{=}
H\Bigg(\sum_{(u,v)\in[U]\times [V]} W_{u,v}
\,\Bigg|\,
\left\{Y^{(1)}_u\right\}_{u\in[U]},
\left\{Y^{(2)}_u\right\}_{u\in[U]\setminus\{u'\}}
\Bigg) \\
&\geq
H\Bigg(\sum_{(u,v)\in[U]\times [V]} W_{u,v}
\,\Bigg|\,
\left\{Y^{(1)}_u\right\}_{u\in[U]},
\left\{Y^{(2)}_u\right\}_{u\in[U]\setminus\{u'\}},
\Big\{X^{(1)}_{u,v}\Big\}_{(u,v)\in([U]\setminus\{u'\})\times [V]}, \notag\\
&\hspace{2.5cm}
\Big\{X^{(2)}_{u,v}\Big\}_{(u,v)\in([U]\setminus\{u'\})\times [V]},
\left\{W_{u,v}, Z_{u,v}\right\}_{(u,v)\in([U]\setminus\{u'\})\times [V]}
\Bigg) \\
&\overset{(\ref{messageX2}),(\ref{messageY1}),(\ref{messageY2})}{=}
H\Bigg( \sum_{(u,v)\in\{(u',v)\}_{v\in [V]}} W_{u,v} \Bigg| Y^{(1)}_{u'},\left\{W_{u,v}, Z_{u,v}\right\}_{(u,v)\in([U]\setminus\{u'\})\times [V]}\Bigg),\label{covpftt2}
\end{align}
where~\eqref{covpftt2} follows since
$\{X^{(1)}_{u,v}\}_{(u,v)\in([U]\setminus\{u'\})\times [V]}$ and
$\{X^{(2)}_{u,v}\}_{(u,v)\in([U]\setminus\{u'\})\times [V]}$ 
are deterministic functions of
$\{W_{u,v}, Z_{u,v}\}_{(u,v)\in([U]\setminus\{u'\})\times [V]}$,
as defined in~(\ref{messageX1}) and~(\ref{messageX2}), and
$\{Y^{(1)}_u\}_{u\in[U]\setminus\{u'\}}$ and
$\{Y^{(2)}_u\}_{u\in[U]\setminus\{u'\}}$ 
are deterministic functions of
$\{X^{(1)}_{u,v}\}_{(u,v)\in([U]\setminus\{u'\})\times [V]}$ and
$\{X^{(2)}_{u,v}\}_{(u,v)\in([U]\setminus\{u'\})\times [V]}$, respectively,
as defined in~(\ref{messageY1}) and~(\ref{messageY2}).
Next,
\begin{eqnarray}
L &\overset{(\ref{inputsize})}{=}& H\Bigg( \sum_{(u,v)\in\{(u',v)\}_{v\in [V]}} W_{u,v}\Bigg)\\
&\overset{(\ref{independent})}{=}& H\left(\sum_{(u,v)\in\{(u',v)\}_{v\in [V]}} W_{u,v} \Bigg|  \left\{W_{u,v},Z_{u,v}\right\}_{(u,v)\in([U]\backslash\{u'\})\times [V]} \right) \\
&=& I\left(\sum_{(u,v)\in\{(u',v)\}_{v\in [V]}} W_{u,v} ; Y^{(1)}_{u'}\Bigg| \left\{W_{u,v},Z_{u,v}\right\}_{(u,v)\in([U]\backslash\{u'\})\times [V]} \right) \notag\\
&&+\underbrace{H\Bigg( \sum_{(u,v)\in\{(u',v)\}_{v\in [V]}} W_{u,v} \Bigg| Y^{(1)}_{u'},\left\{W_{u,v}, Z_{u,v}\right\}_{(u,v)\in([U]\setminus\{u'\})\times [V]}\Bigg)}_{\overset{(\ref{covpftt2})}{=}0}\\
&\leq&  H\left( Y^{(1)}_{u'}\Bigg| \left\{W_{u,v},Z_{u,v}\right\}_{(u,v)\in([U]\backslash\{u'\})\times [V]} \right)  \\ 
&\leq& H\left(Y^{(1)}_{u'} \right) ~\leq~ L^{(1)}_Y ,\\
& \Rightarrow~& R_1 \overset{(\ref{rate})}{=} \frac{L^{(1)}_Y}{L} \geq 1.
\end{eqnarray}

\subsection{Converse for $R^{(2)}_X \geq \frac{1}{V_{0}U_{0}-T}$ When $U_{0}V_{0} > T$} \label{RX2RY2}

We prove the converse bound for the second-round message rate.
Consider the setting where
$\mathcal{V}^{(1)}_u=\mathcal{V}^{(2)}_u=[V_{0}]$,
$\mathcal{U}^{(1)} = [U_{0}+1]$,
$\mathcal{U}^{(2)} = [U_{0}]$,
and let $\mathcal{T} \subset [U_{0}]\times [V_{0}]$ with $|\mathcal{T}|=T$.

Before presenting the formal converse proof, we first outline the main intuition.
The key observation is that, due to the server security constraint, the collection consisting of all first-round second-hop messages together with any $T$ second-round first-hop messages is statistically independent of the desired sum.
Therefore, these messages cannot contribute any useful information for decoding the sum at the server.
As a result, all information required to recover
$\sum_{(u,v)\in[U_0]\times[V_0]} W_{u,v}$
from the first hop must be conveyed by the remaining $U_0V_0 - T$ second-round messages.
Since the entropy of the desired sum equals $L$, it follows that, on average, each of these remaining messages must carry at least $L/(U_0V_0 - T)$ symbols of information.

The following proof formalizes this intuition using standard information-theoretic inequalities.
Specifically, from the server security constraint in \eqref{serversecurity}, we have
\begin{flalign}
    0&\overset{(\ref{serversecurity})}{=} I\Big(\left\{W_{u,v}\right\}_{(u,v)\in[U]\times [V]}; \left\{Y^{(1)}_{u}\right\}_{u\in[U]}, \left\{Y_{u}^{(2)}\right\}_{u\in[U_{0}+1]} \Big| \sum_{(u,v)\in[U_{0}+1]\times [V_{0}]} W_{u,v}, \left\{ W_{u,v}, Z_{u,v} \right\}_{(u,v)\in\mathcal{T}} \Big)& \notag\\
 &\geq I\Bigg(\sum_{(u,v)\in[U_{0}]\times [V_{0}]}W_{u,v}; \left\{Y^{(1)}_{u}\right\}_{u\in[U_{0}]} \Bigg| \sum_{(u,v)\in[U_{0}+1]\times [V_{0}]} W_{u,v}, \left\{ W_{u,v}, Z_{u,v} \right\}_{(u,v)\in\mathcal{T}} \Bigg)& \\
 &\overset{(\ref{messageX2})}{=} I\Bigg(\sum_{(u,v)\in[U_{0}]\times [V_{0}]}W_{u,v}; \left\{Y^{(1)}_{u}\right\}_{u\in[U_{0}]},\left\{ X^{(2)}_{u,v}\right\}_{(u,v)\in\mathcal{T}} \Bigg| \sum_{(u,v)\in[U_{0}+1]\times [V_{0}]} W_{u,v}, \left\{ W_{u,v}, Z_{u,v} \right\}_{(u,v)\in\mathcal{T}} \Bigg)& \\
 &= I\Bigg(\sum_{(u,v)\in[U_{0}]\times [V_{0}]}W_{u,v}; \left\{Y^{(1)}_{u}\right\}_{u\in[U_{0}]},\left\{ X^{(2)}_{u,v}\right\}_{(u,v)\in\mathcal{T}} , \sum_{(u,v)\in[U_{0}+1]\times [V_{0}]} W_{u,v}, \left\{ W_{u,v}, Z_{u,v} \right\}_{(u,v)\in\mathcal{T}} \Bigg)& \notag\\
 &- \underbrace{I\Bigg(\sum_{(u,v)\in[U_{0}]\times [V_{0}]}W_{u,v};  \sum_{(u,v)\in[U_{0}+1]\times [V_{0}]} W_{u,v}, \left\{ W_{u,v}, Z_{u,v} \right\}_{(u,v)\in\mathcal{T}} \Bigg)}_{\overset{(\ref{eq:indc})}{=}0}& \label{covpfr2t1}\\
 &\geq I\Bigg(\sum_{(u,v)\in[U_{0}]\times [V_{0}]}W_{u,v}; \left\{Y^{(1)}_{u}\right\}_{u\in[U_{0}]},\left\{ X^{(2)}_{u,v}\right\}_{(u,v)\in\mathcal{T}}  \Bigg), \label{covpfr2t2}&
\end{flalign}
where the second term of~(\ref{covpfr2t1}) is zero by applying Lemma~\ref{lemma:indc}, with
$\mathcal{T} = [U_1]\times [V_1]$, $V_2 = V_3 = V_0$, $U_2 = U_0$, and $U_3 = U_0 + 1$.

Next, consider $\mathcal{V}^{(1)} = \mathcal{V}^{(2)}  = [V_{0}]$ and $\mathcal{U}^{(1)}  = \mathcal{U}^{(2)}  = [U_{0}]$. From the correctness constraint (\ref{correctness}), we have
\begin{align}
L \overset{(\ref{inputsize})}{=}&  I\left(\sum_{(u,v)\in[U_{0}]\times [V_{0}]}W_{u,v} \right) \label{covpfr2xt0}\\
 =& \underbrace{H\left(\sum_{(u,v)\in[U_{0}]\times [V_{0}]}W_{u,v} \Bigg|\left\{Y^{(1)}_{u}\right\}_{u\in[U_{0}]}, \left\{Y^{(2)}_{u}\right\}_{u\in[U_{0}]} \right)}_{\overset{(\ref{correctness})}{=}0} \notag\\
 &+  I\left(\sum_{(u,v)\in[U_{0}]\times [V_{0}]}W_{u,v} ;\left\{Y^{(1)}_{u}\right\}_{u\in[U_{0}]}, \left\{Y^{(2)}_{u}\right\}_{u\in[U_{0}]} \right) \label{covpfr2xt11}\\
\leq&  I\left(\sum_{(u,v)\in[U_{0}]\times [V_{0}]}W_{u,v} ;\left\{Y^{(1)}_{u}\right\}_{u\in[U_{0}]}, \left\{Y^{(2)}_{u}\right\}_{u\in[U_{0}]},\left\{X^{(2)}_{u,v}\right\}_{(u,v)\in[U_{0}]\times [V_{0}]} \right) \label{covpfr2xt1}\\
\overset{(\ref{messageY2})}{=}&  I\left(\sum_{(u,v)\in[U_{0}]\times [V_{0}]}W_{u,v} ;\left\{Y^{(1)}_{u}\right\}_{u\in[U_{0}]}, \left\{X^{(2)}_{u,v}\right\}_{(u,v)\in[U_{0}]\times [V_{0}]} \right) \label{covpfr2xt2}\\
=& I\left(\sum_{(u,v)\in[U_{0}]\times [V_{0}]}W_{u,v} ; \left\{X^{(2)}_{u,v}\right\}_{(u,v)\in([U_{0}]\times [V_{0}])\backslash \mathcal{T}}\Bigg|\left\{Y^{(1)}_{u}\right\}_{u\in[U_{0}]},\left\{ X^{(2)}_{u,v}\right\}_{(u,v)\in\mathcal{T}}  \right) \notag\\
&+\underbrace{I\Bigg(\sum_{(u,v)\in[U_{0}]\times [V_{0}]}W_{u,v}; \left\{Y^{(1)}_{u}\right\}_{u\in[U_{0}]},\left\{ X^{(2)}_{u,v}\right\}_{(u,v)\in\mathcal{T}}  \Bigg)}_{\overset{(\ref{covpfr2t2})}{=}0}\\
\leq&H\left( \left\{X^{(2)}_{u,v}\right\}_{(u,v)\in([U_{0}]\times [V_{0}])\backslash \mathcal{T}} \right)\\
\leq& \sum_{(u,v)\in([U_{0}]\times [V_{0}])\backslash \mathcal{T}} H\left(X^{(2)}_{u,v}\right)\leq (V_{0}U_{0}-T)L^{(2)}_X, \\
\Rightarrow &  R_2 \overset{(\ref{rate})}{=} \frac{L^{(2)}_X}{L} \geq \frac{1}{V_{0}U_{0}-T}.
\end{align}

We next justify the steps in the above derivation.
(\ref{covpfr2xt0}) follows from the definition of the input size in~(\ref{inputsize}).
The conditional entropy term in the first term of (\ref{covpfr2xt11}) is zero due to the correctness constraint in~(\ref{correctness}),
since $\mathcal{V}^{(1)}_u=\mathcal{V}^{(2)}_u=\{(u,v)\}_{v\in[V_0]}$ and
$\mathcal{U}^{(1)}=\mathcal{U}^{(2)}=[U_0]$, which guarantees that
$\sum_{(u,v)\in[U_0]\times[V_0]} W_{u,v}$ can be recovered from
$\{Y^{(1)}_u\}_{u\in[U_0]}$ and $\{Y^{(2)}_u\}_{u\in[U_0]}$.
(\ref{covpfr2xt1}) follows since adding side information cannot decrease mutual information.
Equality~(\ref{covpfr2xt2}) holds because $\{Y^{(2)}_u\}_{u\in[U_0]}$ is a deterministic function of
$\{X^{(2)}_{u,v}\}_{(u,v)\in[U_0]\times[V_0]}$ according to~(\ref{messageY2}).

\subsection{Lower bound for $R^{(2)}_Y \geq \frac{1}{U_{0}-\lfloor {T}/{V_{0}}  \rfloor}$ When $U_{0}V_{0} > T$} \label{RX2RY22}

We prove the converse bound for the second-round message rate.  
Consider the setting
$\mathcal{V}^{(1)}_u = \mathcal{V}^{(2)}_u = [V_{0}]$, 
$\mathcal{U}^{(1)} = [U_{0}+1]$, 
$\mathcal{U}^{(2)} = [U_{0}]$, and let 
$[\lfloor T/V_0 \rfloor] \times [V_0] \subseteq \mathcal{T}$ with $|\mathcal{T}| = T$.

Before presenting the formal lower bound proof, we briefly explain the main idea.  
Due to the server security constraint, once the server is given the messages and keys corresponding to any set $\mathcal{T}$ of size $T$, the desired sum becomes statistically independent of all first-round second-hop messages and of a subset of second-round second-hop messages.  
In particular, by choosing $\mathcal{T}$ such that $[\lfloor T/V_0 \rfloor] \times [V_0] \subseteq \mathcal{T}$, the second-round messages $\{Y_u^{(2)}\}_{u \in [\lfloor T/V_0 \rfloor]}$ carry no useful information for decoding the desired sum.

Consequently, in the second hop, only the remaining $U_0 - \lfloor T/V_0 \rfloor$ second-round messages can convey information about $\sum_{(u,v) \in [U_0] \times [V_0]} W_{u,v}$.  
Since the entropy of the desired sum equals $L$, this yields the conservative lower bound
$R^{(2)}_Y \ge L / (U_0 - \lfloor T/V_0 \rfloor)$.  
We remark that this bound may not be tight, as it follows from a worst-case integer partition of $\mathcal{T}$ across users; a potentially tighter bound of the form $L / (U_0 - T/V_0)$ is suggested by symmetry, but establishing such a result would require fundamentally different techniques.

Then, from the server security constraint \eqref{serversecurity}, we have
\begin{flalign}
0 \overset{(\ref{serversecurity})}{=}& I\Bigg(\left\{W_{u,v}\right\}_{(u,v)\in[U]\times [V]}; \left\{Y^{(1)}_{u}\right\}_{u\in[U]}, \left\{Y_{u}^{(2)}\right\}_{u\in[U_{0}+1]} \Bigg| \sum_{(u,v)\in[U_{0}+1]\times [V_{0}]} W_{u,v}, \left\{ W_{u,v}, Z_{u,v} \right\}_{(u,v)\in\mathcal{T}} \Bigg)& \notag\\
\geq& I\Big(\sum_{(u,v)\in[U_{0}]\times [V_{0}]}W_{u,v}; \left\{Y^{(1)}_{u}\right\}_{u\in[U_{0}]} \Big| \sum_{(u,v)\in[U_{0}+1]\times [V_{0}]} W_{u,v}, \left\{ W_{u,v}, Z_{u,v} \right\}_{(u,v)\in\mathcal{T}} \Big)& \\
\overset{(\ref{messageY2})}{=}& I\Big(\sum_{(u,v)\in[U_{0}]\times [V_{0}]}W_{u,v}; \left\{Y^{(1)}_{u}\right\}_{u\in[U_{0}]},\left\{ Y^{(2)}_{u}\right\}_{u\in\big[\big\lfloor \tfrac{T}{V_{0}} \big\rfloor\big]} \Big| \sum_{(u,v)\in[U_{0}+1]\times [V_{0}]} W_{u,v}, \left\{ W_{u,v}, Z_{u,v} \right\}_{(u,v)\in\mathcal{T}} \Big)& \label{covpfr2Yt1}\\
=& I\Big(\sum_{(u,v)\in[U_{0}]\times [V_{0}]}W_{u,v}; \left\{Y^{(1)}_{u}\right\}_{u\in[U_{0}]},\left\{ Y^{(2)}_{u}\right\}_{u\in \big[\big\lfloor \tfrac{T}{V_{0}} \big\rfloor\big]} , \sum_{(u,v)\in[U_{0}+1]\times [V_{0}]} W_{u,v}, \left\{ W_{u,v}, Z_{u,v} \right\}_{(u,v)\in\mathcal{T}} \Big)& \notag\\
&- \underbrace{I\Big(\sum_{(u,v)\in[U_{0}]\times [V_{0}]}W_{u,v};  \sum_{(u,v)\in[U_{0}+1]\times [V_{0}]} W_{u,v}, \left\{ W_{u,v}, Z_{u,v} \right\}_{(u,v)\in\mathcal{T}} \Big)}_{\overset{(\ref{eq:indc})}{=}0}& \label{covpfr2t3}\\
\geq&I\Big(\sum_{(u,v)\in[U_{0}]\times [V_{0}]}W_{u,v}; \left\{Y^{(1)}_{u}\right\}_{u\in[U_{0}]},\left\{Y^{(2)}_{u}\right\}_{u\in\big[\big\lfloor \tfrac{T}{V_{0}} \big\rfloor\big]} \Big), \label{covpfr2t22}&
\end{flalign}
where (\ref{covpfr2Yt1}) holds because $[\lfloor T/V_0 \rfloor] \times [V_0] \subseteq \mathcal{T}$.
Hence, for all $u \in \big[\big\lfloor \tfrac{T}{V_0} \big\rfloor\big]$, the second-round symbols
$\{X^{(2)}_{u,v}\}_{v\in[V_0]}$ are completely determined by
$\{W_{u,v}, Z_{u,v}\}_{(u,v)\in\mathcal{T}}$, which are already conditioned upon.
As a result, the corresponding second-round messages
$\{Y^{(2)}_u\}_{u\in\big[\big\lfloor \tfrac{T}{V_0} \big\rfloor\big]}$
are deterministic functions of the conditioning variables and can be added to the mutual information
without affecting its value.
The second term in (\ref{covpfr2t3}) is equal to zero by Lemma~\ref{lemma:indc}.
Specifically, by setting $\mathcal{T} = [U_1]\times[V_1]$ with
$V_2 = V_3 = V_0$, $U_2 = U_0$, and $U_3 = U_0 + 1$,
the desired sum $\sum_{(u,v)\in[U_0]\times[V_0]} W_{u,v}$ is independent of
$\sum_{(u,v)\in[U_0+1]\times[V_0]} W_{u,v}$ and of
$\{W_{u,v}, Z_{u,v}\}_{(u,v)\in\mathcal{T}}$.

Next, consider $\mathcal{V}^{(1)} = \mathcal{V}^{(2)}  = [V_{0}]$ and $\mathcal{U}^{(1)}  = \mathcal{U}^{(2)}  = [U_{0}]$. From the correctness constraint (\ref{correctness}), we have
\begin{align}
L \overset{(\ref{correctness})}{=}&  H\left(\sum_{(u,v)\in[U_{0}]\times [V_{0}]}W_{u,v}  \right) \\
  =& \underbrace{H\left(\sum_{(u,v)\in[U_{0}]\times [V_{0}]}W_{u,v} \Bigg|\left\{Y^{(1)}_{u}\right\}_{u\in[U_{0}]}, \left\{Y^{(2)}_{u}\right\}_{u\in[U_{0}]} \right)}_{\overset{(\ref{correctness})}{=}0} \\
  & + I\left(\sum_{(u,v)\in[U_{0}]\times [V_{0}]}W_{u,v} ;\left\{Y^{(1)}_{u}\right\}_{u\in[U_{0}]}, \left\{Y^{(2)}_{u}\right\}_{u\in[U_{0}]} \right) \label{covpfr2y2t1}\\
=&\underbrace{H\Bigg(\sum_{(u,v)\in[U_{0}]\times [V_{0}]}W_{u,v}; \left\{Y^{(1)}_{u}\right\}_{u\in[U_{0}]},\left\{ Y^{(2)}_{u}\right\}_{u\in\big[\big\lfloor \tfrac{T}{V_{0}} \big\rfloor\big]} \Bigg)}_{\overset{(\ref{covpfr2t22})}{=}0} \notag\\
& + I\left(\sum_{(u,v)\in[U_{0}]\times [V_{0}]}W_{u,v} ; \left\{Y^{(2)}_{u}\right\}_{u\in[U_{0}]\backslash \big[\big\lfloor \tfrac{T}{V_{0}} \big\rfloor\big]}\Bigg|\left\{Y^{(1)}_{u}\right\}_{u\in[U_{0}]},\left\{ Y^{(2)}_{u}\right\}_{u\in\big[\big\lfloor \tfrac{T}{V_{0}} \big\rfloor\big]}  \right) \label{covpfr2y2t2}\\
\leq&H\left( \left\{Y^{(2)}_{u}\right\}_{u\in[U_{0}]\backslash \big[\big\lfloor \tfrac{T}{V_{0}} \big\rfloor\big]} \right)\\
\leq& \sum_{u\in[U_{0}]\backslash \big[\big\lfloor \tfrac{T}{V_{0}} \big\rfloor\big]} H\left(Y^{(2)}_{u}\right)\leq (U_{0}-\big\lfloor \tfrac{T}{V_{0}} \big\rfloor)L^{(2)}_Y,\\
& \Rightarrow R^{(2)}_Y \overset{(\ref{rate})}{=} \frac{L^{(2)}_Y}{L} \geq \frac{1}{U_{0}-\big\lfloor \tfrac{T}{V_{0}} \big\rfloor}.
\end{align}
where the first term  in (\ref{covpfr2y2t1}) is zero due to the correctness constraint in~(\ref{correctness}).
Specifically, under the setting
$\mathcal{V}^{(1)}_u=\mathcal{V}^{(2)}_u=\{(u,v)\}_{v\in[V_0]}$ and
$\mathcal{U}^{(1)}=\mathcal{U}^{(2)}=[U_0]$,
the desired sum $\sum_{(u,v)\in[U_0]\times[V_0]} W_{u,v}$ can be reliably recovered from
$\{Y^{(1)}_u\}_{u\in[U_0]}$ and $\{Y^{(2)}_u\}_{u\in[U_0]}$.
Hence, conditioning on these messages leaves no residual uncertainty.

\section{Conclusion}
\label{sec:conclusion}

In this paper, we studied information-theoretic hierarchical secure aggregation under user and relay dropouts with collusion constraints. We established correctness and security guarantees in a two-round hierarchical model and characterized the communication cost for most links. While tight results were obtained in several regimes, a gap remains in the second-round relay-to-server communication.
Several directions remain open. A primary question is to close the gap in the second-round communication. It is also of interest to extend the model to structured collusion settings, such as group-wise security constraints. Another important direction is to characterize the optimal key rate and understand the tradeoff between shared randomness and communication under dropout. 
Further extensions include heterogeneous network settings with asymmetric user distributions and more realistic dropout models. These directions may provide a deeper understanding of the fundamental limits of hierarchical secure aggregation.

\bibliographystyle{IEEEtran}
\bibliography{references_secagg.bib}
\end{document}